\documentclass[%
 reprint,
superscriptaddress,
 amsmath,amssymb,
 aps,
]{revtex4-2}

\usepackage{graphicx}
\usepackage{dcolumn}
\usepackage{bm}

\usepackage[utf8]{inputenc}
\usepackage[english]{babel}
\usepackage[T1]{fontenc}
\usepackage{amsmath}
\usepackage{amsthm}
\usepackage{amssymb}
\usepackage{hyperref}
\hypersetup{
    colorlinks=true,
    linkcolor=blue,
    filecolor=blue,      
    urlcolor=blue,
    citecolor=blue,
}
\usepackage{array}
\usepackage{booktabs}
\usepackage{CJKutf8}

\usepackage{tikz}
\usepackage{lipsum}

\newtheorem{theorem}{Theorem}

\newcommand{\ket}[1]{ | {#1} \rangle}

\DeclareUnicodeCharacter{2009}{ }
\renewcommand{\selectlanguage}[1]{}

\begin{document}

\preprint{APS/123-QED}

\title{Bosonic Cyclic Codes: Trading Stabilizers for Gaussian Non-Clifford Phase Gates}

\author{Owen C. Wetherbee}
\email[Contact author: ]{owenweth@berkeley.edu}
\affiliation{Department of Physics, Cornell University, Ithaca, NY, 14853, USA}
\affiliation{Department of Physics, University of California, Berkeley, California 94720, USA}

\author{Yijia Xu (\begin{CJK*}{UTF8}{gbsn}许逸葭\end{CJK*})}
\affiliation{Joint Center for Quantum Information and Computer Science, University of Maryland, College Park, Maryland 20742, USA}
\affiliation{Institute for Physical Science and Technology, University of Maryland, College Park, Maryland 20742, USA}

\author{Victor V. Albert}
\affiliation{Joint Center for Quantum Information and Computer Science, NIST/University of Maryland, College Park, Maryland 20742, USA}

\author{Baptiste Royer}
\affiliation{Département de Physique and Institut Quantique, Université de Sherbrooke, Sherbrooke J1K 2R1, QC, Canada}

\author{Valla Fatemi}
\affiliation{School of Applied and Engineering Physics, Cornell University, Ithaca, NY, 14853, USA}

\date{\today}

\begin{abstract}
Bosonic codes offer hardware-efficient approaches to quantum error correction, with the best encodings offering effective protection of idle quantum information against loss and dephasing -- particularly rotation-symmetric codes, which include the cat and binomial code families.
However, rotation-symmetric codes are only naturally endowed with a single logical Pauli gate, while other logical gates require the use of non-linear operations, obstructing the utility of these codes for realizing quantum algorithms.
Here, we balance error protection with controllability by introducing bosonic cyclic codes: a generalization of rotation-symmetric codes that enable the measured tradeoff of error protection properties for fault-tolerant logical phase gates.
Through our general construction, we find that sacrificing the detectability of a single photon loss relative to a rotation-symmetric code can yield a number of logical phase gates commensurate with the original rotation symmetry order of the code, all achievable via passive Gaussian rotations.
Giving the corresponding generalizations of cat and binomial codes -- which we dub cyclic cat and Vandermonde codes, respectively -- we further find that many of the desirable properties of these codes transfer to the bosonic cyclic code setting.
We go on to discuss the larger $SU(2)$ symmetry and rotation gates of the codes, which yield additional stabilizers and logical Pauli gates, as well as new non-Clifford gates for the smallest `kitten' binomial code, and provide a new error detection protocol.
Finally, we introduce a general paradigm for converting higher-order stabilizers to logical gates, as in our generalization of rotation-symmetric codes, and apply it to several multimode bosonic codes.

\end{abstract}

\maketitle

\section{Introduction}
Encoding quantum information in the infinite Hilbert space of a bosonic system provides a hardware-efficient alternative to quantum error correction relative to conventional qubit-based codes~\cite{cai_bosonic_2021,terhal_towards_2020,liu_hybrid_2026}.
There are numerous experimental platforms that can realize such bosonic encodings, such as the electromagnetic modes in optical systems~\cite{kimble_strong_1998,mabuchi_cavity_2002} and the phononic modes in trapped ion systems~\cite{leibfried_quantum_2003,monroe_scaling_2013,fluhmann_encoding_2019}.
Chief among them, superconducting circuits have emerged as a leading candidate due to their low decoherence and large noise bias to photon loss~\cite{reagor_quantum_2016,lei_high_2020,milul_superconducting_2023,ganjam_surpassing_2024}.
In tandem with these hardware platforms, various bosonic encoding schemes have been developed that protect against the dominant errors in these systems, most notably photon loss~\cite{chuang_bosonic_1997,gottesman_encoding_2001,cochrane_macroscopically_1999}.
These include the cat~\cite{cochrane_macroscopically_1999,mirrahimi_dynamically_2014} and binomial~\cite{michael_new_2016} code families, which have enjoyed quick experimental progress, with several realizations recently achieving the so-called break-even point~\cite{hu_quantum_2019,ni_beating_2023,sivak_real-time_2023,ofek_extending_2016}.
The cat and binomial codes can both be understood as specific cases of the more general rotation-symmetric codes~\cite{grimsmo_quantum_2020}, characterized by discrete phase-space rotational symmetry and modular occupation in Fock space.

While these bosonic encodings have shown promise as quantum memories, enhancing the lifetime of idle quantum information, implementing the fault-tolerant operations required for enacting quantum algorithms -- the ultimate goal of quantum computing -- has remained a challenge.
While logical $\bar{Z}$ gates come naturally with the rotation-symmetric construction, implementing any other logical gate requires introducing non-linearities into the system, typically by coupling the bosonic mode to a finite-level ancilla.
These, often noisy, ancillae greatly reduce the effectiveness of the error-correction, as protected quantum information leaks out through the ancilla during the operations.
In addition, as the logical states temporarily leave their encoded subspace during operations, even pure bosonic errors are no longer necessarily protected.

One emerging approach to address this problem is to design codes with logical gates in mind from the outset.
In particular, rather than constructing codes to protect against the dominant errors in a system, and then searching for suitable operations, we can construct codes for which the natural operations of the system enact logical gates, and then later ensure protection against errors.
This methodology has seen great success for constructing big spin~\cite{gross_designing_2021} and multimode bosonic~\cite{jain_quantum_2024,aydin_quantum_2025,wxl_tessellation} encodings, and has blossomed into the so-called covariant encoding formalism~\cite{denys_quantum_2024}.
Of course, opposite to the conventional method of code construction, there is the inherent drawback of a lowered emphasis from the outset on protection against errors.

Here, we take a step towards fault-tolerant control of bosonic encodings, interpolating between these two code-construction methodologies, by introducing bosonic cyclic codes: a generalization of rotation-symmetric codes that enable the measured tradeoff of error protection properties for fault-tolerant logical phase gates.
In particular, we show that by lowering the Fock spacing of a rotation-symmetric code, but retaining the modular spacing of the individual code words, the phase-space rotation stabilizers can be converted to discrete logical phase gates.
This allows for crucial non-Clifford gates, such as $\bar{T}$, to be implemented with reference oscillator phase updates, while retaining the protection of the code against photon losses.
After detailing the general construction, we provide the corresponding generalizations of cat and binomial codes, which we dub cyclic cat and Vandermonde codes, respectively.
We find that many of the desirable properties of these popular rotation-symmetric codes -- e.g., methods of cat code state preparation~\cite{liu_hybrid_2026} and the exactly-correcting nature of binomial codes~\cite{michael_new_2016} -- transfer over to the general bosonic cyclic code case.
We then discuss the larger $SU(2)$ symmetry and rotation gates of the codes, which yield additional stabilizers and logical $\bar{X}$ gates, as well as new non-Clifford gates for the smallest `kitten' binomial code.
Finally, we provide a potential error detection protocol for bosonic cyclic codes based on nested Ramsey interferometry measurements~\cite{sun_tracking_2014,jin_general_2025}, and present a few extensions of our codes to the multimode case.
This multimode extension in particular highlights how our general methodology could be applied to other established encodings, in both bosonic contexts and elsewhere, to realize an analogous tradeoff of error protection properties for natural gates.

\begin{figure*}[ht]
    \centering
    \includegraphics[width=\textwidth]{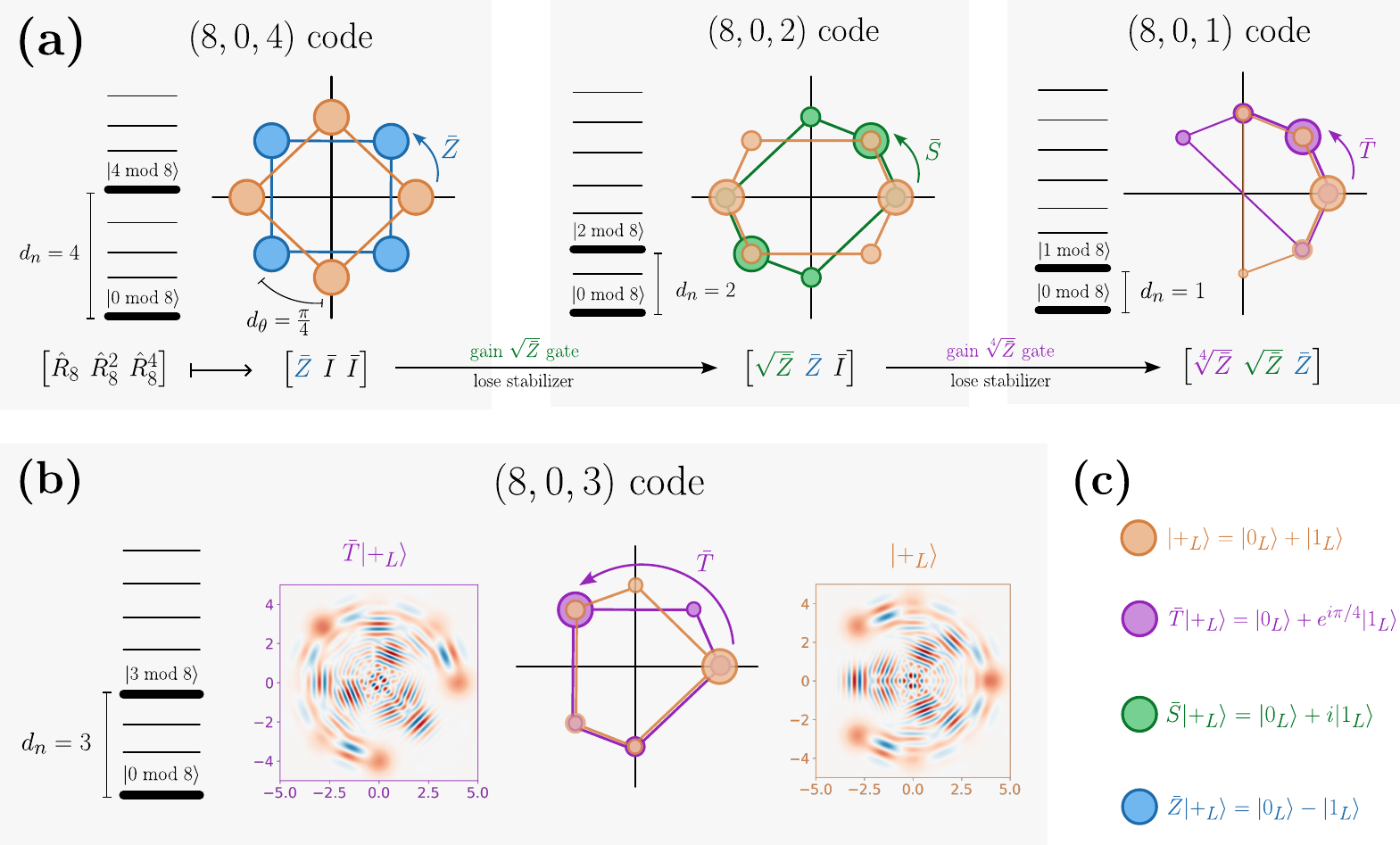}
    \caption{\textbf{Summary of bosonic cyclic codes. (a)} Three bosonic cyclic codes which demonstrate the tradeoff between phase-space rotation gates and codespace stabilizers.
    The three code parameters $(S, n_{0}, n_{1})$ determine the Fock state periodicity, the $|0_{L}\rangle$ Fock support (mod $S$), and the $|1_{L}\rangle$ Fock support (mod $S$), respectively.
    The resulting Fock structure of each code is shown on the left of each panel.
    Ball-and-stick phase-space diagrams of relevant logical states (see (c)) are depicted on the right of each panel to demonstrate the rotation symmetry and gates of each code.
    The balls correspond to peaks of the Wigner function, with their size indicating the relative amplitude of the peak, and the sticks correspond to fringes.
    At the bottom of each panel, the number of available rotation gates is indicated by the logical gate ($\bar{Z}$, $\bar{S} = \sqrt{\bar{Z}}$, or $\bar{T} = \sqrt[4]{\bar{Z}}$) achieved by a $\hat{R}_{S} = \hat{R}_{8} = \exp\left(i\frac{\pi}{4}\hat{n}\right)$ rotation.
    We note that these three codes can also be constructed as traditional rotation-symmetric codes.
    \textbf{(b)} A summary of the $(8, 0, 3)$ cyclic code -- notably not a rotation-symmetric code -- which better balances the above tradeoff by having a phase-space rotation $\bar{T}$ gate while still achieving a code distance of $d_{n} = 3$.
    The Wigner functions of the logical states in (c) for a $(8, 0, 3)$ cyclic cat code with $\alpha = 4$ are also shown alongside the ball-and-stick diagram, clearly demonstrating the $\bar{T}$ action of the $\hat{R}_{8}^{3} = \exp\left(i\frac{3\pi}{4}\hat{n}\right)$ rotation.
    \textbf{(c)} Logical states represented in the ball-and-stick diagrams in (a) and (b).}
    \label{fig:cyclicCodesSummary}
\end{figure*}

\section{Bosonic cyclic codes}\label{sec:bosonicCyclicCodes}
The bosonic cyclic codes we define here generalize bosonic rotation-symmetric codes to allow for many more `easy' gates at the expense of slightly reduced protection against photon loss.
As with rotation-symmetric codes~\cite{grimsmo_quantum_2020}, these cyclic codes can be naturally defined either in terms of their rotation-symmetry or their Fock-space structure.
Beginning with the former, we define an order-$S$ cyclic code to be any code where the discrete rotation operator
\begin{equation}
    \begin{aligned}
        \hat{R}_{S} = \exp\left(i\frac{2\pi}{S}\hat{n}\right)
    \end{aligned}
    \label{eqn:rotSymmetry}
\end{equation}
acts as a nontrivial logical $\bar{Z}(\theta)$ gate on the codespace\footnote{By convention, we require $\hat{R}_{S}$ to be the smallest rotation which acts as a logical $\bar{Z}(\theta)$ on the codespace, so that $S$ is the maximum order of the code.}, where $\hat{n} = \hat{a}^{\dagger}\hat{a}$ is the photon number operator.
Note that since $\hat{R}_{S}^{S} = \hat{I}$, $\theta$ must be of the form $\theta = \frac{2\pi N}{S}$ for some $N \neq 0 \in \mathbb{Z}_{S}$.

Similar to rotation-symmetric codes, any order-$S$ cyclic code can be constructed from rotated superpositions of some primitive state $|\Theta\rangle$\footnote{To yield nonzero codewords, $|\Theta\rangle$ must have support on at least one of the $|n_{0} \bmod S\rangle$ Fock states and one of the $|n_{1} \bmod S\rangle$ Fock states.}:
\begin{equation}
    \begin{aligned}
        |0_{L}\rangle & = \frac{1}{\sqrt{\mathcal{N}_{0}}}\sum_{m = 0}^{S - 1}\left[ \hat R_S e^{i\frac{2\pi}{S} n_{0}}\right]^m|\Theta\rangle ~, \\
        |1_{L}\rangle & = \frac{1}{\sqrt{\mathcal{N}_{1}}}\sum_{m = 0}^{S - 1}\left[ \hat R_S e^{i\frac{2\pi}{S} n_{1}}\right]^m|\Theta\rangle ~,
    \end{aligned}
    \label{eqn:phaseStructure}
\end{equation}
where $\mathcal{N}_{0}$ and $\mathcal{N}_{1}$ are normalization constants and $n_{0} \neq n_{1} \in \mathbb{Z}_{S}$.
The logical codewords of a code can be written in this form if and only if $\hat{R}_{S}$ acts as a logical $\bar{Z}\left(\frac{2\pi N}{S}\right)$ gate on the code for $N = (n_{1} - n_{0}) \bmod S$.
In this way, a bosonic cyclic code is entirely specified by a primitive state $|\Theta\rangle$ and the code parameters $S$, $n_{0}$, and $n_{1}$.

A bosonic cyclic code can also be specified by its Fock-grid coefficients $c_{k}$.
Indeed, once $S$, $n_{0}$, and $n_{1}$ are given, the Fock-space structure of a cyclic code is constrained to be of the form
\begin{equation}
    \begin{aligned}
        |0_{L}\rangle & = \sum_{j = 0}^{\infty}c_{2j}|jS + n_{0}\rangle ~, \\
        |1_{L}\rangle & = \sum_{j = 0}^{\infty}c_{2j + 1}|jS + n_{1}\rangle ~,
    \end{aligned}
    \label{eqn:fockStructure}
\end{equation}
where the coefficients $c_{k}$ satisfy the normalization condition
\begin{equation}
    \sum_{j = 0}^{\infty}|c_{2j}|^{2} = \sum_{j = 0}^{\infty}|c_{2j + 1}|^{2} = 1 ~.
\end{equation}

From these definitions, it is clear that order-$N$ rotation-symmetric codes are the special case of these newly defined cyclic codes that arise when $n_{0} = 0$, $n_{1} = N$, and $S = 2N$~\cite{grimsmo_quantum_2020}.
In this case, $\hat{R}_{2N}$ acts as the usual $\bar{Z}\left(\frac{2\pi(N - 0)}{2N}\right) = \bar{Z}$ gate and both codewords are fully supported on the $|0 \bmod N\rangle$ Fock states.
We remark that the order $N$ of a rotation-symmetric code is always half of what we call the order $S$ of the corresponding cyclic code.

Both of $N$ and $S$ have a straightforward interpretation in terms of the Fock-space structure of Eq.~\eqref{eqn:fockStructure}.
$S$ corresponds to the \emph{word spacing} between adjacent occupied Fock states of each codeword and $N = (n_{1} - n_{0}) \bmod S$ corresponds to the minimum \emph{code spacing} between adjacent occupied Fock states of the codespace.
For cyclic codes, the code spacing is technically given by $\text{min}\left(N, S - N\right)$, but we generally take $N \leq S / 2$ so that this is $N$.

\subsection{Detectable errors and phase gates}\label{sec:errorGates}
The dual phase and Fock space structure of a generic $(S, n_{0}, n_{1})$ cyclic code endows it with the ability to detect a certain amount of photon losses and phase space rotations, specified by the number distance $d_{n} = \min(N, S - N)$ and rotation distance $d_{\theta} = 2\pi / S$, respectively.
As with rotation-symmetric codes, the order $S$ of the code represents a tradeoff between these two error distances: as $S$ increases, the Fock spacing of the code increases while the rotation distance decreases.

The advantage of considering bosonic cyclic codes is that, by fixing $S$ and varying $N$, an additional tradeoff is introduced between the Fock spacing of the code and the number of phase gates achievable via phase-space rotations.
Indeed, when $N$ is set to $N = S / 2$ to maximize Fock spacing, as in rotation-symmetric codes, Eq.~\eqref{eqn:phaseStructure} shows that $\hat{R}_{S}$ acts as a Pauli $\bar{Z}(\pi) = \bar{Z}$ gate, and since $\bar{Z}^{2} = \bar{I}$, this is the only nontrivial gate achievable via phase-space rotations.
Decreasing $N$ to $N = S / 4$, $\hat{R}_{S}$ instead acts as a phase $\bar{Z}\left(\frac{\pi}{2}\right) = \bar{S}$ gate, endowing the code with two additional phase-space rotation gates, but with the Fock spacing decreased by a factor of two.
Notably, this latter code is still technically a rotation-symmetric code, just with the order halved to $S / 4$ and many of the Fock-grid coefficients periodically set to zero.
This `spacing-gates tradeoff' can be visualized by placing the $|n \bmod S\rangle$ Fock states evenly around the unit circle of the complex plane, with $|0 \bmod S\rangle$ located at $e^{i(0)} = 1$, as in Fig.~\ref{fig:fockPhase} for $S = 8$.
The placement of these Fock states corresponds to the phase they pick up under a $\hat{R}_{S}$ rotation.
Indicating the $|n_{0} \bmod S\rangle$ and $|n_{1} \bmod S\rangle$ Fock states on this diagram, the angle between them then equals the relative phase $\theta$ of the $\bar{Z}(\theta)$ gate arising from the $\hat{R}_{S}$ rotation, and the acute angle between them represents the Fock spacing of the code.
Since a photon loss preserves the relative Fock spacing of the codewords -- shifting their placements on the Fock-phase diagram clockwise by one rung -- the $\hat{R}_{S}$ rotation has the same $\bar{Z}(\theta)$ action on the error spaces as the code space.
We remark that, while the angle between adjacent Fock states does correspond to the rotation distance of the code, these diagrams should not be thought of as visually representing the code in phase space.

\begin{figure}[ht]
    \centering
    \includegraphics[width=0.45\textwidth]{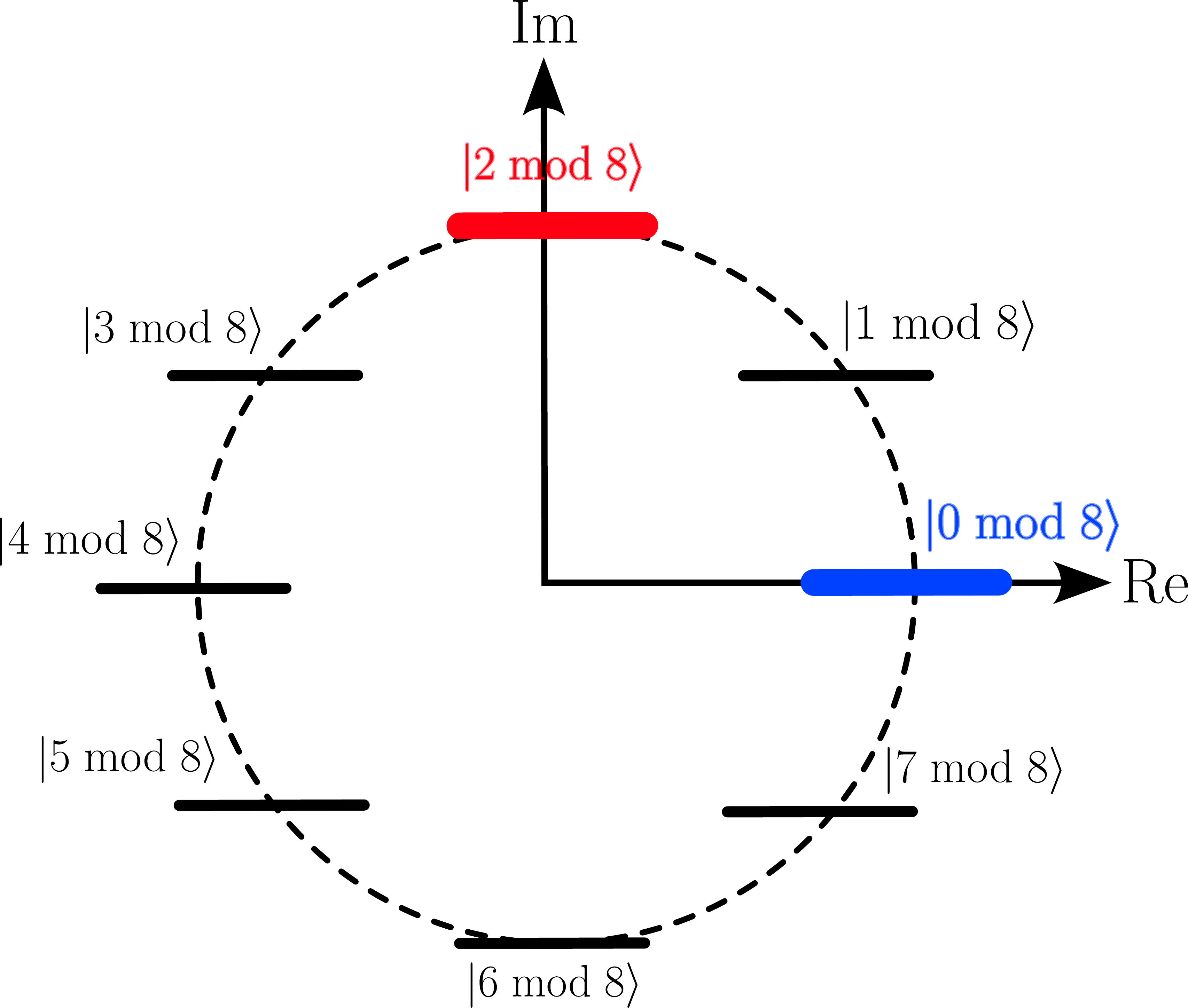}
    \caption{\textbf{Fock-phase diagram} for a $(S, n_{0}, n_{1}) = (8, 0, 2)$ code, representing the phase picked up by each $|n \bmod S\rangle$ parity manifold under a $\hat{R}_{S}$ rotation.
    This visualization elucidates that this code has a number distance of $d_{n} = 2$ and phase gate $\bar{S} = \bar{Z}\left(\frac{\pi}{2}\right)$ given by the phase-space rotation $\hat{R}_{8} = \exp\left(i\frac{\pi}{4}\hat{n}\right)$.
    Conveniently, the rotation distance of the code is also represented by the angle between the evenly-spaced Fock states, which in this case is $d_{\theta} = \frac{\pi}{4}$.
    }
    \label{fig:fockPhase}
\end{figure}

This visualization makes it clear that halving $N$, while more straightforward intuitively, does not make efficient use of the spacing-gates tradeoff.
Indeed, in the $S = 8$ case of Fig.~\ref{fig:fockPhase}, if we instead set $n_{1} = 3$, then $\hat{R}_{8}$ yields the logical phase gate $\bar{Z}\left(\frac{3\pi}{4}\right) = \bar{T}^{3}$.
Applying this rotation thrice then yields a logical $\bar{T}^{9} = \bar{T}$ gate, so that eventually all powers of $\bar{T}$ gates are generated by the $\hat{R}_{8}$ rotation.
Thus, while only losing one Fock spacing relative to the traditional $n_{1} = 4$ rotation-symmetric code, this code gains $\bar{S}$ and $\bar{T}$ gates, both achievable by simple phase-space rotations.
This hinges on the fact that, while $4$ and $8$ share a large common factor, $3$ and $8$ are coprime, so that $\left[\bar{Z}\left(2\pi\left(3 / 8\right)\right)\right]^{k}$ only equals $\bar{I}$ when $k$ is a multiple of $8$.
More generally, the number of phase gates a bosonic cyclic code has that are achievable by phase-space rotations is given by $S / \gcd(S, N)$, so that all the $S$ possible phase gates are achieved exactly when $S$ and $N$ are coprime.
Thus, from the perspective of the spacing-gates tradeoff, $N$ should be selected as the number coprime with $S$ that achieves the largest number distance $d_{n}$.
As summarized in Table~\ref{tab:cyclicCodesCases}, this number distance never has to be decreased more than $2$ from $S / 2$ to achieve an $N$ that is coprime with $S$.
In this way, decreasing the number distance of an order-$N$ rotation-symmetric code by only one or two transforms the $\bar{Z}$ gate to a $\sqrt[N]{\bar{Z}}$ gate, yielding $2N - 2$ additional phase gates achievable by phase-space rotations.

\begin{table}
    \centering
    \def\arraystretch{1.75}
    \newcolumntype{?}{!{\vrule width 0.8pt}}
    \aboverulesep = 0pt
    \belowrulesep = 0pt
    \begin{tabular}{c|ccc}
        \bottomrule
        \multicolumn{1}{r|}{$\Delta d_{n} \to$} & $0$ & $-1$ & $-2$ \\
        \bottomrule
        $S$ odd & $\sqrt[S/2]{\bar{Z}}$ & - & - \\
        $S / 2$ even & $\bar{Z}$ & $\sqrt[S/2]{\bar{Z}}$ & - \\
        $S / 2$ odd & $\bar{Z}$ & $\sqrt[S / 4]{\bar{Z}}$ & $\sqrt[S/2]{\bar{Z}}$ \\
        \bottomrule
    \end{tabular}
    \caption{Summary of the smallest phase gates achieved for order-$S$ cyclic codes with number distances $d_{n}$ in the vicinity of the maximal distance $d_{n} = S / 2$ (or $d_{n} = (S - 1) / 2$ for odd $S$), broken down into three cases according to the parity of $S$. Here, the exponent $S/2$ is not an integer for odd $S$, so $\sqrt[S/2]{\bar{Z}}$ is understood as the fractional power $\bar{Z}^{2/S} = \bar{Z}(2\pi/S)$ rather than an integer root.
    For $S$ odd, we already have $\gcd(S, N) = 1$ for $N = (S - 1) / 2$, so all the possible phase gates are achieved with the maximal distance.
    For $S$ even and $S / 2$ even, $\gcd(S, N) = N$ for the maximal $N = S / 2$, but $\gcd(S, N) = 1$ for $N = S / 2 - 1$, so decreasing the code distance by $1$ ensures all the phase gates are achieved.
    Finally, for $S$ even and $S / 2$ odd, $\gcd(S, N) = 2$ for $N = S / 2 - 1$, in which case half of the phase gates are achieved, so the code distance needs to be decreased by one more to realize all the phase gates.}
    \label{tab:cyclicCodesCases}
\end{table}

\subsection{Cyclic cat codes}\label{sec:cyclicCatCodes}
Our construction provides a natural generalization of rotation-symmetric cat codes when the primitive is chosen to be a coherent state $|\Theta\rangle = |\alpha\rangle$.
In this case, the codewords in Eq.\eqref{eqn:phaseStructure} become equal superpositions of rotations of $|\alpha\rangle$, and can equivalently be written as~\cite{grimsmo_quantum_2020}
\begin{equation}
    \begin{aligned}
        |0/1_{L}\rangle & = \frac{1}{\sqrt{\mathcal{N}_{0/1}}}\hat{\Pi}_{S}^{n_{0/1}}|\alpha\rangle ~,
    \end{aligned}
\end{equation}
where $\mathcal{N}_{0/1}$ are normalization constants and $\hat{\Pi}_{S}^{n} = \sum_{k = 0}^{\infty}|Sk + n\rangle\langle Sk + n|$ is the projector onto the $|n \bmod S\rangle$ Fock states.
When $n_{0} = 0$ and $n_{1} = S / 2$, we arrive at the usual rotation-symmetric cat codes.
Notably, since coherent states are eigenstates of $\hat{a}$, the error words of cat codes take on a similar form as the codewords: $\hat{a}^{m}\left(\hat{\Pi}_{S}^{n}|\alpha\rangle\right) \propto \hat{\Pi}_{S}^{n - m}|\alpha\rangle$.
Thus, adjusting $n_{0}$ and $n_{1}$ to yield more general cyclic codes amounts to replacing one of the codewords of a traditional cat code with one of its error words.

Since cyclic cat codewords are different combinations of rotation-symmetric cat codes and error words, they inherit many properties from traditional cat codes.
For instance, following the analysis in~\cite{li_cat_2017}, cyclic cat codes also have $\alpha$ ``sweetspots'' where photon loss and dephasing errors are suppressed at first order, corresponding to when $\langle 0_{L}|\hat{n}|0_{L}\rangle = \langle 1_{L}|\hat{n}|1_{L}\rangle$.
In addition, methods of state preparation can be readily generalized to cyclic cat codes by projecting onto the appropriate $|n \bmod S\rangle$ manifold.
However, difficulties arise when generalizing error detection strategies to cyclic cat codes since their code and error spaces no longer live on single $|n \bmod S / 2\rangle$ manifolds.
This means generalized parity measurements are no longer sufficient to determine the error syndromes, and more complicated POVMs are required to perform non-destructive error detection.
We discuss this limitation in Sec.~\ref{sec:errorDetection}.

The positive tradeoff is the additional nontrivial phase gates available to cyclic codes.
One example of an $(S, n_{0}, n_{1}) = (8, 0, 3)$ cyclic cat code is shown in Fig.~\ref{fig:cyclicCodesSummary}(b).
As depicted, this code has $\bar{T}$ gates achievable via phase-space rotations, yet retains nearly all of the error-correction properties of the corresponding $(8, 0, 4)$ rotation-symmetric cat code.

\subsection{Vandermonde codes}\label{sec:vand}
We now discuss the subset of cyclic codes that naturally generalize binomial codes~\cite{michael_new_2016}.
The key feature of binomial codes we seek to generalize is that they exactly correct a given number of photon losses using a finite support in Fock space.
Bosonic codes whose $|0_{L}\rangle$ and $|1_{L}\rangle$ codewords have support on disjoint Fock states -- such as binomial, rotation-symmetric, and cyclic codes -- exactly correct the first $l$ photon losses if and only if the $\hat{n}^{l'}$ moments of $|0_{L}\rangle$ and $|1_{L}\rangle$ are equal up to $\hat{n}^{l}$.
In other words,
\begin{equation}
  \begin{aligned}
    \langle 0_{L} | \hat{n}^{l'} | 0_{L} \rangle & = \langle 1_{L} | \hat{n}^{l'} | 1_{L} \rangle \hspace{10pt} \forall\,\,l' \in \{0, \dots, l\} ~, \\
  \end{aligned}
  \label{eqn:equalMoments}
\end{equation}
where the $l' = 0$ case follows from codeword normalization.
The binomial coefficients, for which binomial codes are named, arise when taking $|0_{L}\rangle$ and $|1_{L}\rangle$ as a finite-support rotation-symmetric code and solving for the Fock-grid coefficients which satisfy Eq.~\eqref{eqn:equalMoments}.
However, this derivation relies on the constant Fock spacing of rotation-symmetric codes.
For more general Fock spacings, like those of cyclic codes, we must instead employ the following theorem.
\begin{theorem}
    A code with codewords $|0_{L}\rangle$ and $|1_{L}\rangle$, having finite disjoint alternating Fock support so that
    \begin{equation}
        \begin{aligned}
            |0/1_{L}\rangle & = \sum_{k \text{ even/odd}}^{K}c_{k}|n_{k}\rangle \\
        \end{aligned}
        \label{eqn:thm1Codewords}
    \end{equation}
    for Fock-grid coefficients $c_{0}, c_{1}, \dots, c_{K}$ and ascending Fock state support $n_{0} < n_{1} < \dots < n_{K}$, satisfies Eq.~\eqref{eqn:equalMoments} up to $l = K - 1$ if and only if
    \begin{equation}
        \begin{aligned}
            |c_{k}|^{2} = \frac{1}{\prod_{i \neq k}|n_{k} - n_{i}|}
        \end{aligned}
        \label{eqn:thm1Coeffs}
    \end{equation}
    for all $k$, up to some overall scaling factor.
    \label{thm:vandermonde}
\end{theorem}
\begin{proof}
    Plugging the codewords from Eq.~\eqref{eqn:thm1Codewords} into Eq.~\eqref{eqn:equalMoments}, the equal-moments conditions become
    \begin{equation}
      \begin{aligned}
        \sum_{k \text{ even}}^{K}c_{k}^{*}n_{k}^{l'}c_{k} & = \sum_{k \text{ odd}}^{K}c_{k}^{*}n_{k}^{l'}c_{k}  \\
        \iff \sum_{k = 0}^{K}(-1)^{k}n_{k}^{l'}|c_{k}|^{2} & = 0 ~,
      \end{aligned}
      \label{eqn:squaredCoeffsCond}
    \end{equation}
    for all $l' \in \{0, \dots, l\}$.
    These equations are linear in the squared coefficients $|c_{k}|^{2}$, and can be written in matrix form as $V\vec{b} = \vec{0}$, where $\vec{b} = \begin{pmatrix}|c_{0}|^{2} & -|c_{1}|^{2} & |c_{2}|^{2} \dots & (-1)^{K}|c_{K}|^{2}\end{pmatrix}^{\top}$ is the vector of alternating-sign squared coefficients and $V$ is the $(l + 1) \times (K + 1)$ Vandermonde matrix~\cite{macon_inverses_1958}
    \begin{equation}
        \begin{aligned}
            V = \begin{pmatrix}
                1 & 1 & 1 & \cdots & 1 \\
                n_{0} & n_{1} & n_{2} & \cdots & n_{K} \\
                n_{0}^{2} & n_{1}^{2} & n_{2}^{2} & \cdots & n_{K}^{2} \\
                \vdots & \vdots & \vdots & \ddots & \vdots \\
                n_{0}^{l} & n_{1}^{l} & n_{2}^{l} & \cdots & n_{K}^{l}
            \end{pmatrix} ~.
        \end{aligned}
        \label{eqn:vandermondeMat}
    \end{equation}
    When $l = K - 1$, we can construct a square Vandermonde matrix $V'$ from $V$ by adding one extra row of $n_{k}^{l + 1}$ powers.
    Since the $n_{k}$ are distinct, $V'$ is known to be invertible~\cite{macon_inverses_1958}, and by definition, the last column of $(V')^{-1}$ is orthogonal to all but the last row of $V'$.
    Therefore, the kernel of $V$, and thus the set of $\vec{b}$ satisfying $V\vec{b} = \vec{0}$, is spanned by the last column of $(V')^{-1}$, which is known to be given by the explicit formula~\cite{macon_inverses_1958}
    \begin{equation}
        \begin{aligned}
            (V')^{-1}_{k, K} = \prod_{i \neq k}\frac{1}{n_{k} - n_{i}} ~.
        \end{aligned}
    \end{equation}
    Since the $n_{k}$ are increasing, $(V')^{-1}_{k, K}$ alternates in sign so that any solution $\vec{b}$ has coefficients
    \begin{equation}
        \begin{aligned}
            b_{k} = \prod_{i \neq k}\frac{(-1)^{k}}{|n_{k} - n_{i}|} ~,
        \end{aligned}
        \label{eqn:bkCoeffs}
    \end{equation}
    up to some overall scaling factor.
    Plugging in our definition $b_{k} = (-1)^{k}|c_{k}|^{2}$ then concludes the proof.
\end{proof}

Using this theorem, we define Vandermonde codes for a given Fock space structure $n_{k}$ and coefficient number $K > 0 \in \mathbb{N}$ as the codes of the form Eq.~\eqref{eqn:thm1Codewords} with coefficients given by Eq.~\eqref{eqn:thm1Coeffs}, arising from the kernel of the Vandermonde matrix in Eq.~\eqref{eqn:vandermondeMat} for which these codes are named.
We usually only focus on cyclic Vandermonde codes, which have periodic Fock support $n_{k \text{ even/odd}} = kS + n_{0/1}$, but Theorem~\ref{thm:vandermonde} applies more generally to any code with disjoint and alternating Fock support.
When the Fock support is evenly spaced so that $n_{k} = kN$, the coefficients in Eq.~\eqref{eqn:thm1Coeffs} are proportional to binomial coefficients (see Appendix~\ref{app:vandCoeffs}), recovering the usual binomial codes.

From the proof of Theorem~\ref{thm:vandermonde}, we see that for a cyclic code to exactly correct the first $l$ photon losses, it must occupy at least $K + 1 = l + 2$ Fock states, since when $K < l + 1$, the Vandermonde matrix in Eq.~\eqref{eqn:vandermondeMat} has trivial kernel and thus there is no solution to Eq.~\eqref{eqn:equalMoments}.
In this way, Vandermonde codes minimize the number of occupied Fock states required to exactly correct for photon loss, at least among codes having codewords with disjoint Fock support.
When more than $l + 2$ Fock states are used, the kernel of the Vandermonde matrix becomes higher dimensional so that there are more solutions than Eq.~\eqref{eqn:thm1Coeffs}.
Analogous to the generalized binomial codes in~\cite{wetherbee_mathematical_2025}, these additional solutions can be constructed from the coefficients of smaller Vandermonde codes.

\section{$SU(2)$ rotation gates}\label{sec:rotGates}
Codes with finite Fock support (e.g., Vandermonde codes) could alternatively be implemented in finite-dimensional spin systems~\cite{yu_schrodinger_2025}, in which case the natural operations are extended from $U(1)$ phase-space rotations to full $SU(2)$ rotations.
In addition, several recent works have shown that synthetic $SU(2)$ spin rotations can be engineered inside a cavity~\cite{roy_synthetic_2025,champion_efficient_2025}, further encouraging us to consider what logical gates can be achieved for bosonic cyclic codes via $SU(2)$ rotations.
In particular, we are interested in finite-support cyclic codes that, when mapped to a $d$-dimensional spin system, are covariant with respect to some subgroup $G \subset SU(2)$.

By \emph{covariant}, we mean that for each $g \in G$, the physical rotation $\lambda(g)$ accomplishes the corresponding logical operation $\mu(g)$, where $\lambda$ is the $d$-dimensional spin representation of $G$ and $\mu$ is some $2$-dimensional representation of $G$~\cite{denys_quantum_2024}.
This problem has already been more or less solved in~\cite{gross_designing_2021} -- specifically for $G$ the binary tetrahedral, binary octahedral ($2O$), and binary icosahedral groups -- but here we emphasize connections to bosonic cyclic codes by additionally considering the dihedral groups and noting that the smallest binomial code in fact realizes an unfaithful irrep of the group $G \cong 2O$.

By definition (Eq.~\ref{eqn:rotSymmetry}), bosonic cyclic codes are already covariant with respect to the cyclic group $G = C_{S} = \{e, g, g^{2}, \dots, g^{S - 1}\}$ for any $d$-dimensional subspace that completely contains the code.
In this case, the $\lambda : C_{S} \to SU(d)$ `spin' representation is given by $\lambda(g^{k}) = \hat{R}_{S}^{k}$.
The $\mu : C_{S} \to SU(2)$ representation depends on the placement of the Fock states according to $\mu(g^{k}) = \bar{Z}\left(\frac{2\pi kN}{S}\right)$, where $N = (n_{1} - n_{0}) \bmod S$, and is faithful exactly when $N$ and $S$ are coprime.
This is another way to understand the discussion in Sec.~\ref{sec:errorGates}.
In fact, we could have originally defined (finite-support) order-$S$ bosonic cyclic codes as codes that are covariant to the cyclic group $C_{S}$, rotated appropriately. 

\subsection{Bosonic dihedral codes}\label{sec:dihedral}
A cyclic group is upgraded to a dihedral group\footnote{These are called binary dihedral groups when the total spin \(j\) is half-integer, but this distinction is not relevant to the encodings since representations of the two groups are the same up to global phase.}
by adding a second involutive group generator $h$ which accomplishes a physical $\lambda(h) = \hat{X}_{d} = \exp\left(-i\pi \hat{J}_{x}^{(d)}\right)$ rotation, where $\hat{J}_{x}^{(d)}$ is the angular momentum operator in the $x$ direction for spin $j = \frac{d - 1}{2}$.
Thus, a bosonic cyclic code is upgraded to a bosonic dihedral code if there is some spin dimension $d$ such that $\hat{X}_{d}$ acts as a logical gate.
By the commutation relations of the group generators $g$ and $h$ (i.e. $\hat{R}_{S}$ and $\hat{X}_{d}$), the logical gate $\mu(h)$ must either be a $\pi$ rotation about an equatorial axis of the Bloch sphere, or else be identity; so $\mu(h) = \bar{X}$ or $\bar{I}$, up to a change of logical basis. 

From the symmetry property $d_{m', m}^{j}(\pi) = (-1)^{j - m}\delta_{m', -m}$ of the small Wigner $D$-matrix~\cite{wigner_gruppentheorie_1931}, the action of the physical rotation $\hat{X}_{d}$ in the Fock basis is given by $\hat{X}_{d}|n\rangle \propto |(d - 1) - n\rangle$, up to an irrelevant $d$-dependent phase.
This operation, depicted in Fig.~\ref{fig:su2Rots}, serves to take the first $d$ Fock states and flip them about their midpoint.
Thus, for $\hat{X}_{d}$ to map the codespace to itself, the Fock-space structure of the code must be invariant under this flip symmetry.
The Fock-grid coefficients must similarly be symmetric about the midpoint index $k = \frac{K}{2}$: $c_{k} = c_{K - k}$ for all $k \in \{0, \dots, K\}$, where $c_{k} = 0$ for $k > K$ so that $K$ is the coefficient index of the highest occupied Fock state, lower than $|d\rangle$, in the code Fock structure.
When these conditions are satisfied, $\hat{X}_{d}$ acts as a logical $\bar{X}$ gate or stabilizer on the codespace.

\begin{figure*}[ht]
    \centering
    \includegraphics[width=\textwidth]{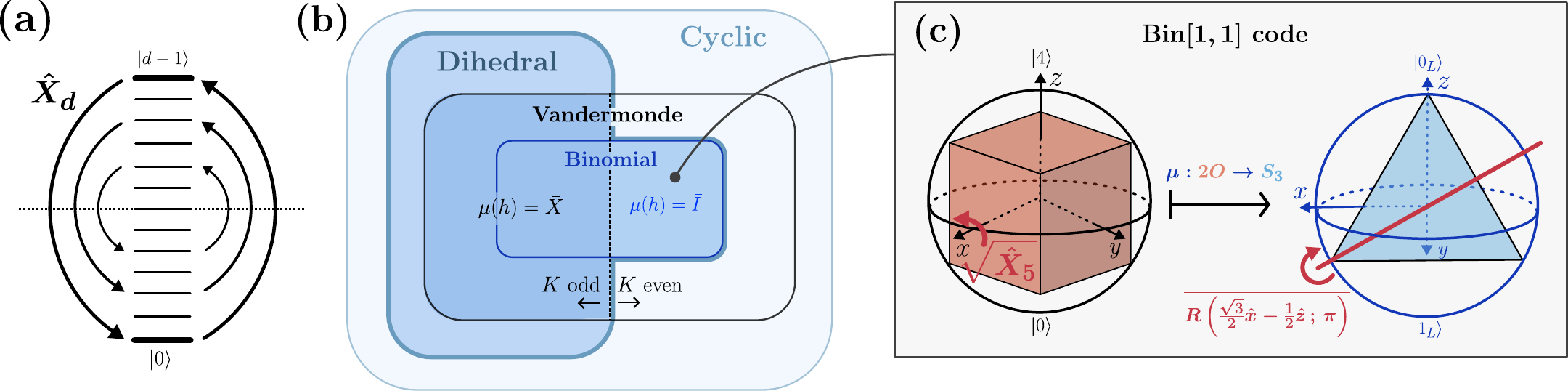}
    \caption{\textbf{Dihedral codes and $\boldsymbol{SU(2)}$ symmetry. (a)} Action of the $\hat{X}_{d} = \exp\left(-i\pi\hat{J}_{x}^{(d)}\right)$ $d$-dimensional spin rotation on the first $d$ Fock states.
    \textbf{(b)} Diagram of the relationship between cyclic, dihedral, Vandermonde, and binomial codes, with the dotted line dividing the latter two codes according to the parity of $K$.
    As indicated, the logical gate $\mu(h)$ corresponding to the extra dihedral group generator $h$ is $\mu(h) = \bar{X}$ for odd $K$ Vandermonde codes and $\mu(h) = \bar{I}$ for even $K$ binomial codes.
    Even $K$ non-binomial Vandermonde codes are not dihedral, so $\mu(h)$ does not have a well-defined logical action in this case.
    \textbf{(c)} Depiction of the octahedral symmetry of the $\text{Bin}[1, 1]$ code for spin dimension $d = 5$.
    The image of the unfaithful logical representation $\mu$ is isomorphic to $\mu(2O) \cong S_{3}$, the symmetry group of the equilateral triangle, and is specifically given by the Bloch sphere rotations that leave the depicted embedded triangle invariant.
    The rotations $\lambda\left(R\left(\hat{x}\,;\,\frac{\pi}{2}\right)\right) = \sqrt{\hat{X}_{5}}$ and $\mu\left(R\left(\hat{x}\,;\,\frac{\pi}{2}\right)\right) = \overline{R\left(\frac{\sqrt{3}}{2}\hat{x} - \frac{1}{2}\hat{z}\,;\,\pi\right)}$ are depicted on the physical dimension-$5$ spin sphere (left) and logical Bloch sphere (right), respectively.}
    \label{fig:su2Rots}
\end{figure*}

Starting with order-$N$ rotation-symmetric codes~\cite{grimsmo_quantum_2020}, the Fock structure condition is always satisfied as long as we choose the spin dimension to be $d = d' + n_{0} = KN + 2n_{0} + 1$, where $d'$ is the smallest spin dimension that fits the code and $c_{K}$ is the highest nonzero coefficient of the code.
Since the binomial distribution is symmetric about its center, binomial codes also always satisfy the Fock-grid coefficient condition.
Thus, all binomial codes are also dihedral codes, having the $SU(2)$ rotation $\hat{X}_{d}$ as a $\bar{X}$ gate or stabilizer.
Specifically, for binomial codes with even $K$, so that $|0_{L}\rangle$ has support on the first and last occupied Fock states, $\hat{X}_{d}$ acts as a stabilizer, sending $|0_{L}\rangle$ and $|1_{L}\rangle$ to themselves.
For odd $K$, $|0_{L}\rangle$ has support on the first occupied Fock state and $|1_{L}\rangle$ has support on the last occupied Fock state, so $\hat{X}_{d}$ instead acts as a logical $\bar{X}$ gate.
In other words, $\mu(h) = (\bar{X})^{K}$ for binomial codes.

Turning to general order-$S$ cyclic codes, the Fock-space structure condition is again always satisfied for odd $K$ as long as we choose spin dimension $d = d' + n_{0} = \frac{1}{2}(K - 1)S + n_{0} + n_{1} + 1$ (see Appendix~\ref{app:vandCoeffs}).
However, for even $K$, it is only satisfied if the code is also rotation-symmetric.
To see this, recall that $\hat{X}_{d}$ acts on the spin-$d$ irrep as the Fock-space reflection $|n\rangle \mapsto |d - 1 - n\rangle$, so the Fock-structure condition requires the code's set of occupied Fock states to be invariant under this reflection. For odd $K$, each logical word has an odd number $K$ of occupied rungs and hence a central rung, so the choice of $d$ above centers the grid and makes its support reflection-symmetric. For even $K$ there is no central rung, and the reflection can only preserve the combined support by mapping the $\{n_{0} + jS\}$ and $\{n_{1} + jS\}$ sublattices of $|0_{L}\rangle$ and $|1_{L}\rangle$ into one another; this requires the two sublattices to be evenly interleaved, $n_{1} - n_{0} = S/2$, which is precisely the rotation-symmetric condition.
Similar to binomial codes, Vandermonde codes with odd $K$ have symmetric coefficients (see Appendix~\ref{app:vandCoeffs}) and thus always satisfy the Fock-grid coefficient condition.
We therefore conclude that all odd $K$ Vandermonde codes are dihedral, with $\mu(h) = \bar{X}$ so that $\hat{X}_{d}$ always acts as a $\bar{X}$ gate.

A summary of how these various cases (e.g., binomial vs. Vandermonde, even $K$ vs. odd $K$) determine whether the code is dihedral and dictate the representation $\mu(h)$ is depicted in Fig.~\ref{fig:su2Rots}(b).

\subsection{Octahedral symmetry of $\text{Bin}[1, 1]$ code}
The $\text{Bin}[1, 1]$ `kitten' code (written using the original parametrization from~\cite{michael_new_2016}) is the smallest, and most common, binomial code, given by:
\begin{equation}
    \begin{aligned}
        |0_{L}\rangle & = \frac{1}{\sqrt{2}}\left(|0\rangle + |4\rangle\right) ~, \\
        |1_{L}\rangle & = |2\rangle ~.
    \end{aligned}
    \label{eqn:bin11}
\end{equation}
In our language, this is a $(S, n_{0}, n_{1}) = (4, 0, 2)$ cyclic (and Vandermonde) code with coefficients $c_{0} = c_{2} = \frac{1}{\sqrt{2}}$, $c_{1} = 1$ and maximum coefficient index $K = 2$.
From Sec.~\ref{sec:dihedral}, we know that the $\text{Bin}[1, 1]$ code is a dihedral code for spin dimension $d = 5$ (i.e. spin $j = 2$) with $\hat{X}_{5}$ a stabilizer.
However, it turns out that this code has further spin-$2$ $SU(2)$ symmetry.

We can see this most directly by writing the spin-$2$ rotation $\sqrt{\hat{X}_{5}} = \exp\left(-i\frac{\pi}{2}\hat{J}_{x}^{(5)}\right)$ explicitly using the Wigner $D$-matrix as~\cite{eden_computer_2003}
\begin{equation*}
    \begin{aligned}
        \sqrt{\hat{X}_{5}} & = \frac{1}{4}\begin{pmatrix} 1 & -2i & -\sqrt{6} & 2i & 1 \\ -2i & -2 & 0 & -2 & 2i \\ -\sqrt{6} & 0 & -2 & 0 &  -\sqrt{6} \\ 2i & -2 & 0 & -2 & -2i \\ 1 & 2i & -\sqrt{6} & -2i & 1 \end{pmatrix}.
    \end{aligned}
\end{equation*}
Applying this to the logical states from Eq.~\eqref{eqn:bin11}, we find
\begin{equation}
    \begin{aligned}
        \sqrt{\hat{X}_{5}}|0_{L}\rangle & = \frac{1}{2}|0_{L}\rangle - \frac{\sqrt{3}}{2}|1_{L}\rangle ~, \\
        \sqrt{\hat{X}_{5}}|1_{L}\rangle & = -\frac{\sqrt{3}}{2}|0_{L}\rangle - \frac{1}{2}|1_{L}\rangle ~,
    \end{aligned}
\end{equation}
revealing that $\sqrt{\hat{X}_{5}}$ acts as a logical gate corresponding to the Bloch sphere rotation of $\pi$ around the $(\theta, \phi) = (2\pi / 3, 0)$ axis, $\overline{R\left(\frac{\sqrt{3}}{2}\hat{x} - \frac{1}{2}\hat{z}\,;\,\pi\right)}$.

Together with the usual $\hat{R}_{4} \mapsto \bar{Z}$ gate of the $\text{Bin}[1, 1]$ code, this rotation generates the spin-$2$ representation $\lambda$ of the octahedral group $G = 2O$ (which, since spin $2$ is an integer-spin irrep, acts through the quotient $2O/\{\pm I\} \cong O \cong S_{4}$), with the corresponding logical gates $\mu$ constituting the symmetries of a triangle inscribed in the Bloch sphere, as depicted in Fig.~\ref{fig:su2Rots}(c).
This is, in fact, precisely the unfaithful two-dimensional irreducible representation of $2O$ mentioned in passing in the supplemental of~\cite{gross_designing_2021}.
The realization that the commonly implemented $\text{Bin}[1, 1]$ code harbors this unfaithful representation of $2O$ is, to the best of our knowledge, novel, and implies this code has non-Clifford triangle-symmetry gates given by $SU(2)$ rotations.
The code is a particular embedding of a more general ``intrinsic'' code into the spin space~\cite{kubischta_intrinsic_2025}, and we clarify its explicit gate set.
This might be particularly attractive for implementations of this code in true big spin systems having native $SU(2)$ operations.

\section{Error detection}\label{sec:errorDetection}
As illustrated in Fig.~\ref{fig:cyclicCodesSummary}, bosonic cyclic codes enable the tradeoff of phase-space rotation stabilizers for phase gates.
Conceptually, this loss of stabilizers is accompanied by a loss of ways to extract error syndromes, so we generically expect error detection to be more difficult.
Indeed, unlike rotation-symmetric codes, cyclic code codespaces are not necessarily contained in a single parity manifold in Fock space, precluding the use of generalized parity measurements for error detection.
Nevertheless, we find that by combining Ramsey interferometry measurement (RIM)~\cite{sun_tracking_2014} with selective parity manifold driving via frequency comb control~\cite{ni_beating_2023} of an ancilla qubit, we can non-destructively detect several photon losses in small Vandermonde codes.

\begin{figure*}[ht]
    \centering
    \includegraphics[width=\textwidth]{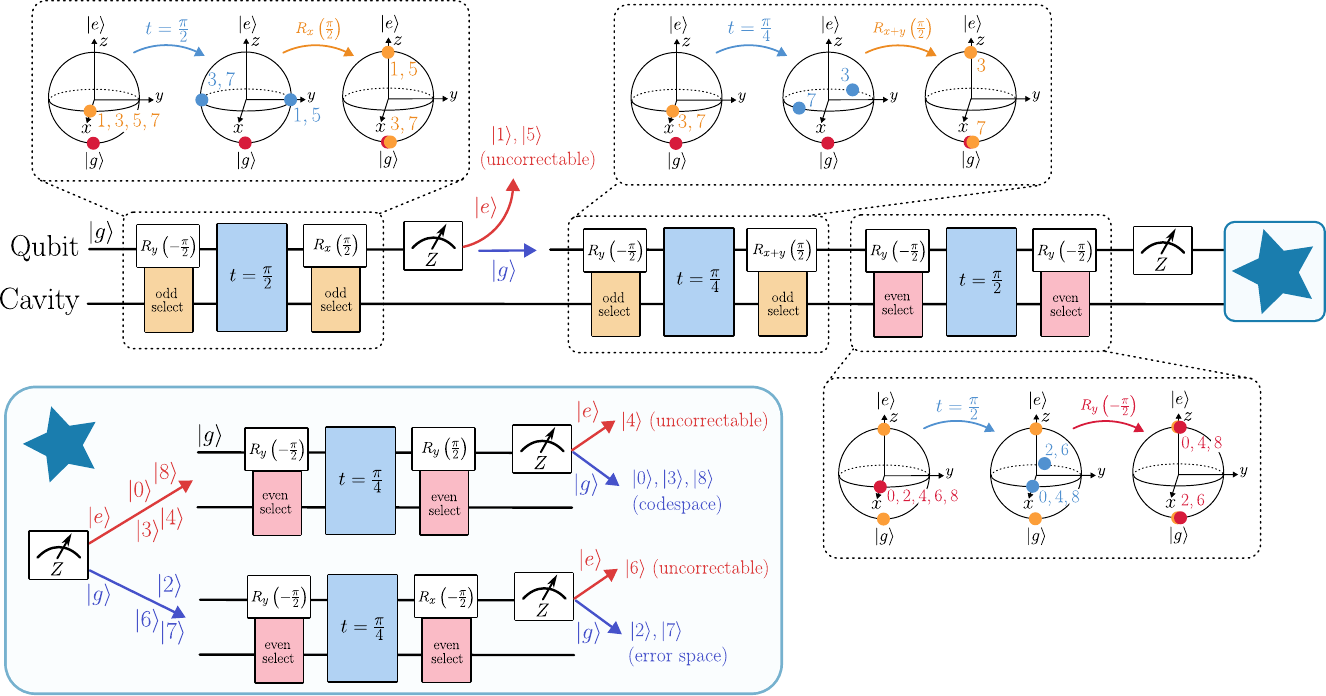}
    \caption{\textbf{Error detection protocol for finite-support $(8, 0, 3)$ code with $K = 2$.}
    The protocol involves a dispersively coupled qubit-cavity system, with the qubit starting in its ground state $|g\rangle$.
    A sequence of adaptive RIMs are then performed, using a combination of idle evolutions $U = \exp\left(it\hat{n}|e\rangle\langle e|\right)$ (blue) and selective qubit rotations for the even (red) and odd (yellow) Fock parity manifolds.
    The evolution of the relevant Fock states around the ancilla Bloch sphere during the first two RIM steps are shown, with the ket notation dropped (i.e. $|j\rangle \to j$).
    A measurement of $|e\rangle$ after the first RIM step implies occupation of the Fock state manifold $\{|1\rangle, |5\rangle\}$, corresponding to an uncorrectable error.
    Following instead a measurement of $|g\rangle$, the second RIM step then distinguishes the Fock state manifolds $\{|0\rangle, |3\rangle, |4\rangle, |8\rangle\}$ and $\{|2\rangle, |6\rangle, |7\rangle\}$, which are then further broken down in the final RIM step into an uncorrectable error $\{|4\rangle\}$ and the codespace $\{|0\rangle, |3\rangle, |8\rangle\}$, and an uncorrectable error $\{|6\rangle\}$ and the error space $\{|2\rangle, |7\rangle\}$, respectively.
    Altogether, this accomplishes a POVM on the relevant Fock states that non-destructively detects if the codespace, error space, or an uncorrectable error space is occupied.
    Note that, as presented, in the second RIM step, the $|3\rangle$ state picks up a phase $e^{i\pi / 4}$ relative to the $|0\rangle$ and $|8\rangle$ states, which can be corrected by adding a $t = \frac{5\pi}{4}$ idle evolution $U = \exp\left(i\frac{5\pi}{4}\hat{n}|e\rangle\langle e|\right)$ before the qubit is reset for the third RIM step.
    }
    \label{fig:errorDetection}
\end{figure*}

Here, we derive the error-detection protocol for the $(S, n_{0}, n_{1}) = (8, 0, 3)$ Vandermonde code with $K = 2$ coefficients, which can correct one photon loss and has highest occupied Fock state $|8\rangle$.
The protocol is illustrated in Fig.~\ref{fig:errorDetection}, and involves achieving the POVM with elements: $\{|0\rangle\langle 0| + |3\rangle\langle 3| + |8\rangle\langle 8|, |2\rangle\langle 2| + |7\rangle\langle 7|, |1\rangle\langle 1| + |5\rangle\langle 5|, |4\rangle\langle 4|, |6\rangle\langle 6|\}$, which successfully distinguishes the codespace and error space from the other potentially occupied Fock states.
This can be accomplished by composing three types of operations on a dispersively coupled qubit-cavity system that houses the bosonic cyclic code in the cavity.
The first type of operation, depicted in blue, is (1) idle evolution of the dispersively coupled system, enacting the unitary $U = \exp\left(it\hat{n}|e\rangle\langle e|\right)$ for some time parameter $t$, where $|e\rangle$ is the excited state of the ancilla qubit.
The other two types of operations, depicted in red and yellow respectively, are selective equatorial rotations of the qubit for only the (2) even or (3) odd Fock states, which can be accomplished, for example, using the frequency comb scheme from Ref.~\cite{ni_beating_2023}.
This frequency comb allows for a faster Rabi rate by including a large number of symmetric detunings around the resonant frequencies of the occupied Fock states, which restricts its use to small codes with finite Fock space support.
For more details, see Appendix.~\ref{app:errorDetection}.

Using these three types of operations, we can perform separate RIMs in the even and odd Fock parity manifolds before measuring the qubit, which is key to building up more complicated POVMs.
In particular, we can perform a sequence of adaptive RIMs, similar to~\cite{jin_general_2025}, where in each step the qubit ends in $|e\rangle$ for a certain subset $\{|m_{1}\rangle, \dots, |m_{J}\rangle\} \subset \{|0\rangle, \dots, |8\rangle\}$ of the relevant Fock states, and ends in $|g\rangle$ otherwise, so that measuring the qubit accomplishes the POVM element $\sum_{j = 1}^{J}|m_{j}\rangle\langle m_{j}|$.
Using the specific adaptive RIM sequence detailed in Fig.~\ref{fig:errorDetection}, we can accomplish the desired POVM.

We give similar derivations for the $(5, 0, 2)$, $K = 2$ Vandermonde code and the $(8, 0, 3)$, $K = 3$ Vandermonde code in Appendix.~\ref{app:errorDetection}.
While it appears that the precise protocol and pulse sequence must be worked out separately for each cyclic code, we expect the strategies used for these three codes to yield corresponding protocols for other small Vandermonde codes.
In addition, while these error-detection protocols do not work exactly for cyclic cat codes due to their infinite Fock support, they could be applied to cat codes with small $\alpha$ for which the higher Fock states are exponentially suppressed.

\section{Stabilizer-to-gate paradigm and extension to multimode bosonic codes}
The way bosonic cyclic codes generalize rotation-symmetric codes by converting the order-$N$ stabilizer $\hat{R}_{N}$ to a phase gate suggests a general paradigm for introducing gate-protection tradeoffs into codes with higher-order stabilizers.
In bosonic cyclic codes, codewords are given by weighted superpositions of rotated primitive states with complex amplitudes.
Moving beyond equal-weight superpositions enables the construction of codes that admit simple implementations of non-Clifford gates, a feature that has also been explored in continuous-variable (CV) systems defined on two-dimensional manifolds \cite{wxl_tessellation}.
We have seen that useful weighted superposition codes can be constructed by starting with a parent code having equal-weight superpositions (e.g., a cat code) and converting a stabilizer into a gate by replacing a codeword with an error word.
In this section, we discuss this general code construction procedure and then apply it by extending an already-established two-mode bosonic code.

Consider a code $\mathcal{C}$, with logical states $|0_{L}\rangle$ and $|1_{L}\rangle$, that detects an error set $\mathcal{E} = \{E_{1}, E_{2}, \dots, E_{N - 1}\}$ via an order-$N$ modular stabilizer $\hat{S}_{N}$, meaning $\hat{S}_{N}^{N} = I$ and the eigenvalues (i.e., syndromes) of $\hat{S}_{N}$ label the error spaces of $E_{1}, \dots, E_{N - 1}$.
In other words, $\hat{S}_{N}|\mu_{L}\rangle = |\mu_{L}\rangle$ and $\hat{S}_{N}E_{k}|\mu_{L}\rangle = e^{2\pi i k / N}E_{k}|\mu_{L}\rangle$ for all $\mu \in \{0, 1\}$ and $k \in \{1, \dots, N - 1\}$, where we generally order the $E_{k}$ according to their $N$-th root of unity syndrome $e^{2\pi ik / N}$.
We can then define a new code $\mathcal{C}'$ by $|0_{L}'\rangle = |0_{L}\rangle$ and $|1_{L}'\rangle \propto E_{k}|1_{L}\rangle$ for some $E_{k} \in \mathcal{E}$, so that $\hat{S}_{N}$ is converted from a stabilizer to a logical $\bar{Z}\left(\frac{2\pi k}{N}\right)$ gate.
Many of the Knill-Laflamme conditions of $\mathcal{C}$ with respect to $\mathcal{E}$ transfer over to $\mathcal{C}'$, so $\mathcal{C}'$ often still detects some non-trivial subset $\mathcal{E}' \subset \mathcal{E}$ of the error set.

In the bosonic cyclic code case, $\mathcal{C}$ can be any order-$N$ rotation-symmetric code, in any computational basis $|0_{L}\rangle$ and $|1_{L}\rangle$, which corrects $\mathcal{E} = \{\hat{a}^{N - 1}, \hat{a}^{N - 2}, \dots, \hat{a}\}$ with corresponding stabilizer $\hat{S}_{N} = \hat{R}_{N} = \exp\left(i\frac{2\pi}{N}\hat{n}\right)$.
Defining new codewords $|0_{L}'\rangle = |0_{L}\rangle$ and $|1_{L}'\rangle \propto \hat{a}|1_{L}\rangle$, we arrive at an order-$N$\footnote{If $|0_{L}\rangle$ and $|1_{L}\rangle$ are the standard computational basis of the order-$N$ rotation-symmetric code $\mathcal{C}$, then the resulting bosonic cyclic code instead has order $2N$ due to the logical $\bar{Z}$ gate action of $\hat{R}_{2N}$ on $\mathcal{C}$ in this basis.} bosonic cyclic code $\mathcal{C}'$ which still detects $\mathcal{E}' = \{\hat{a}^{N - 2}, \dots, \hat{a}\}$ and has $\hat{R}_{N}$ acting as a logical $\bar{Z}\left(\frac{2\pi(N - 1)}{N}\right)$ gate.
Indeed, cyclic cat codes can be defined from regular cat codes precisely in this way, as discussed in Sec.~\ref{sec:cyclicCatCodes}.
Constructing Vandermonde codes from binomial codes in this way also requires adjusting the Fock-grid coefficients, as derived in Sec.~\ref{sec:vand}.
This is a consequence of the fact that $\mathcal{E}'$ will generally not be correctable by $\mathcal{C}'$, even if $\mathcal{E}$ is correctable by $\mathcal{C}$, unless some code parameters are also adjusted.

We now apply this paradigm to the two-mode $\mathbb{Z}_{5}$ simplex code introduced in Ref.~\cite{jain_quantum_2024}.
Here we describe this code and its stabilizer structure at the most basic level; for a more detailed description, see Ref.~\cite{jain_quantum_2024} or Appendix~\ref{app:simplex}.
The two-mode $\mathbb{Z}_{5}$ simplex code $\mathcal{C}$ can be defined in the Fock basis, up to normalization, as
\begin{equation*}
    \begin{aligned}
        |0_{L}\rangle & \propto \sum\limits_{\begin{array}{c} \scriptstyle n, m \geq 0 \\ \scriptstyle (n + 2m) \bmod 5 = 0 \\ \scriptstyle n + m \text{ is even} \end{array}}\frac{\alpha^{n + m}}{\sqrt{n!m!}}|n, m\rangle ~, \\
        |1_{L}\rangle & \propto \sum\limits_{\begin{array}{c} \scriptstyle n, m \geq 0 \\ \scriptstyle (n + 2m) \bmod 5 = 0 \\ \scriptstyle n + m \text{ is odd} \end{array}}\frac{\alpha^{n + m}}{\sqrt{n!m!}}|n, m\rangle ~,
    \end{aligned}
\end{equation*}
for some $\alpha \in \mathbb{C}$.
This code is stabilized by $\hat{S}_{5} = (e^{i 2\pi/5})^{\hat{a}_{1}^{\dagger}\hat{a}_{1} + 2\hat{a}_{2}^{\dagger}\hat{a}_{2}}$, whose syndromes $\alpha e^{2\pi i k / 5}$ for $k \in \{1, \dots, 4\}$ allow the correction of the error set $\mathcal{E} = \{\hat{a}_{2}^{2}, \hat{a}_{1}\hat{a}_{2}, \hat{a}_{2}, \hat{a}_{1}\}$.
Following the discussion above, we then define a new code $\mathcal{C}'$ via $|0_{L}'\rangle = |0_{L}\rangle$ and $|1_{L}'\rangle \propto \hat{a}_{2}|1_{L}\rangle$.
This new code has a logical $\bar{Z}\left(\frac{6\pi}{5}\right)$ gate given by the previous stabilizer $\hat{S}_{5}$ -- a Gaussian operation that generates the full order-5 group of $\bar{Z}(2\pi/5)$ phase gates -- and still detects single photon losses in either mode, $\mathcal{E}' = \{\hat{a}_{1}, \hat{a}_{2}\}$.
We discuss this example, as well as another from the multimode bosonic context, in more detail in Appendix.~\ref{app:multimode}.

In this way, we see how the general stabilizer-to-gate paradigm can be applied to existing codes to construct new codes that have additional easily-implementable gates while retaining some error-protection properties, thus leveraging the gate-protection tradeoff.

\section{Discussion}
We have introduced bosonic cyclic codes, a generalization of bosonic rotation-symmetric codes which are constructed by converting phase-space rotation stabilizers to logical phase gates, including non-Clifford logical phase gates.
The cyclic codes can retain most of their protection against photon loss, and the logical phase gates are relatively easy to implement.
Our results show that many desirable properties of established encodings are preserved under this generalization.
Specifically, cyclic cat codes inherit the exactly-correcting sweetspots and state-preparation techniques of their rotation-symmetric predecessors~\cite{li_cat_2017}, while Vandermonde codes generalize binomial codes~\cite{michael_new_2016} to exactly correct photon loss, even when the code's Fock-state support is not evenly spaced.
In addition to enriching the theoretical structure of rotation-symmetric code families, these generalized codes could be effective when integrated into existing bosonic quantum computing architectures~\cite{liu_hybrid_2026} for use where a high-duty-cycle of operations is expected.
For instance, a computing system could utilize conventional rotation-symmetric codes for stable quantum memory and switch to cyclic variants when rapid, sequence-heavy gates are required.

Especially relevant for true big spin implementations~\cite{yu_schrodinger_2025}, we also discussed the larger $SU(2)$ symmetry of the finite-support bosonic cyclic codes.
We found logical $\bar{X}$ gates for a subset of the Vandermonde codes and revealed a previously unappreciated octahedral symmetry in the $\text{Bin}[1, 1]$ `kitten' code, which endows it with non-Clifford gates given by $SU(2)$ rotations.

We next turned to error detection, which becomes more difficult with general bosonic cyclic codes due to the codewords no longer necessarily occupying a single parity subspace.
Despite this difficulty, we showed that any number of photon losses can still be detected in small Vandermonde codes when nested Ramsey interferometry measurement~\cite{sun_tracking_2014,jin_general_2025} is combined with selective parity manifold driving~\cite{ni_beating_2023}.
While we only presented protocols for three specific finite-support codes, the novel combination of techniques we used to produce the necessary POVMs should work similarly for many small bosonic codes.
These explicit POVM constructions could also potentially have use in other contexts beyond error detection.
Future works in this direction could involve combining our error detection methods with optimal control to produce faster and higher-fidelity recovery operations.
Alternatively, recent work on nonlinear reservoir engineering~\cite{rojkov_stabilization_2026} could be applied to the cyclic code case to explore the possibility of stabilizing these more general cat-state manifolds.

Finally, we expanded the relationship between bosonic cyclic and rotation-symmetric codes into a general paradigm for converting higher-order stabilizers of codes to easily implementable gates, while retaining most of their error protection properties.
We then applied this paradigm to the two-mode simplex code of Ref.~\cite{jain_quantum_2024}, converting it to an analogous code which has an additional Gaussian logical $\bar{Z}\left(\frac{2\pi}{5}\right)$ gate while still detecting single photon losses.
This paradigm could be applied to other established multimode codes, and perhaps even to non-bosonic codes, as long as they have a higher-order stabilizer that labels some error set.
The paradigm could also be used to guide the construction of new codes that readily take advantage of this theoretical structure.
Exploring this stabilizer-to-gate paradigm more generally could be one promising avenue of further investigation.

Our work explores a new explicit tradeoff between error protection and logical control in bosonic codes, introducing practical considerations for the design and selection of error-correcting codes.
For example, consider bosonic codes that correct one photon loss.
A standard choice would be the 0-2-4 kitten code ($(S, n_{0}, n_{1}) = (4, 0, 2)$ with $K = 2$; codewords given in Eq.~\eqref{eqn:bin11}), which corrects one photon loss and has a logical $\bar{Z}$ gate via Gaussian rotation.
The kitten code is minimal in average energy, thereby minimizing physical error rates.
Our work suggests another option: we can instead start with a code that corrects two photon losses and then use our stabilizer-to-gate construction to reduce correction to only one photon loss, gaining a potentially non-Clifford gate.
A simple example would be the 0-2-5 Vandermonde code ($(S, n_{0}, n_{1}) = (5, 0, 2)$ with $K = 2$; codewords given in Eq.~\eqref{eqn:025words} in Appendix~\ref{app:errorDetection}), which corrects one photon loss and has a non-Clifford order-$5$ logical phase gate $\bar{Z}(2\pi/5)$ via Gaussian rotation.
On the one hand, this new code will have modestly higher physical error rates relative to the kitten code due to its higher average photon number (in fact, the logical code space has the same average photon number as the kitten code, leading to no increase of rate for the first error), along with more complicated error syndromes and recovery operations.
On the other hand, this code has a logical non-Clifford gate that can be implemented as a simple `software update,' which can be fast, high-fidelity, and error-transparent~\cite{ma_error-transparent_2020}.
Combining this with a single logical $H$ gate -- ideally also error-transparent~\cite{wetherbee_mathematical_2025,roy_error_2026} -- could enable the construction of general gates through Solovay-Kitaev style construction that are robust to photon loss, opening up a path towards fault-tolerant universal control of a bosonic logical qubit.
Determining which of these options is preferred in a given situation -- a code that minimizes physical error rates and enables simpler recovery, or a code that enables more fault-tolerant gates -- undoubtedly involves weighing many practical and experimental factors, a full consideration of which would be an important direction for future work.

\begin{acknowledgements}
We gratefully acknowledge the support of the Aref and Manon Lahham Faculty Fellowship for support of O. C. W.. 
B.R. acknowledges support of the Natural Sciences and Engineering Research Council of Canada (NSERC) and from Fonds de Recherche du Québec - Nature et Technologie (FRQNT).
V.V.A. acknowledges NSF grant OMA2120757 (QLCI).
V.F. acknowledges support by the NSF under award number 2512537.

\end{acknowledgements}

\textbf{Author Contributions:}
O. C. W. performed the research, with supervision and guidance from V. F. and B. R.. Y. X. and V. V. A. developed the extension to multimode codes, and Y. X. wrote App.~\ref{app:multimode} with input from all authors. O. C. W. wrote the manuscript with oversight from V. F. and input from B. R..

\onecolumngrid
\appendix

\section{Vandermonde code coefficients and dihedral symmetry}\label{app:vandCoeffs}
Here, we prove several properties of Vandermonde code coefficients and Fock-space structures that are referenced in the main text.
Recall that the Vandemonde code coefficients $c_{0}, c_{1}, \dots, c_{K}$ are defined, up to phase, in Eq.~\eqref{eqn:thm1Coeffs} as
\begin{equation}
    \begin{aligned}
        |c_{k}|^{2} & = \frac{1}{\prod_{i \neq k}|n_{k} - n_{i}|} ~,
    \end{aligned}
    \label{eqn:vandCoeffs}
\end{equation}
for ascending Fock support $n_{0} < n_{1} < \dots < n_{K}$.

As mentioned in Sec.~\ref{sec:vand}, when the Fock support is evenly spaced so that $n_{k} = kN$, the squared coefficients in Eq.~\eqref{eqn:vandCoeffs} are proportional to binomial coefficients:
\begin{equation}
    \begin{aligned}
        |c_{k}|^{2} & = \frac{1}{\prod_{i \neq k}|kN - iN|} \\
        & = \frac{1}{N^{K}\left(\prod_{i = 0}^{k - 1}(k - i)\right)\left(\prod_{i = k + 1}^{K}(i - k)\right)} \\
        & = \frac{1}{N^{K}}\frac{1}{\left(\prod_{j = 1}^{k}j\right)\left(\prod_{j = 1}^{K - k}j\right)} \\
        & = \frac{1}{N^{K}K!}\frac{K!}{k!(K - k)!} \\
        & = \frac{1}{N^{K}K!}\binom{K}{k} ~.
    \end{aligned}
\end{equation}

As discussed in Sec.~\ref{sec:dihedral}, the Fock-space structure of a $(S, n_{0}, n_{1})$ cyclic code with odd $K$ is invariant under the flip symmetry $\hat{X}_{d}$ for spin dimension $d = \frac{1}{2}(K - 1)S + n_{0} + n_{1} + 1$.
Specifically, the action of $\hat{X}_{d}$, given by $\hat{X}_{d}|n\rangle = |(d - 1) - n\rangle$, swaps the Fock support of $|0_{L}\rangle$ and $|1_{L}\rangle$ when $K$ is odd.
This follows from the Fock-space structure of Eq.~\ref{eqn:fockStructure}, since for odd $K$ and even $k$ we have
\begin{equation}
    \begin{aligned}
        \hat{X}_{d}\left|\frac{k}{2}S + n_{0}\right\rangle & = \left|\left(\frac{1}{2}(K - 1)S + n_{0} + n_{1}\right) - \left(\frac{k}{2}S + n_{0}\right)\right\rangle \\
        & = \left|\frac{(K - k) - 1}{2}S + n_{1}\right\rangle ~,
    \end{aligned}
\end{equation}
and for odd $k$ we have
\begin{equation}
    \begin{aligned}
        \hat{X}_{d}\left|\frac{k - 1}{2}S + n_{1}\right\rangle & = \left|\left(\frac{1}{2}(K - 1)S + n_{0} + n_{1}\right) - \left(\frac{k - 1}{2}S + n_{1}\right)\right\rangle \\
        & = \left|\frac{K - k}{2}S + n_{0}\right\rangle ~.
    \end{aligned}
\end{equation}
Thus, $\hat{X}_{d}$ sends the Fock state of the $k$-th coefficient, which is in the support of $|0_{L}\rangle$ ($|1_{L}\rangle$) for $k$ even (odd), to the Fock state of the $(K - k)$-th coefficient, which is in the support of $|1_{L}\rangle$ ($|0_{L}\rangle$) for $k$ odd (even), as claimed.

Since the Vandermonde code coefficients in Eq.~\eqref{eqn:vandCoeffs} only depend on the relative spacing between pairs of occupied Fock states, this symmetry of the Fock-space structure implies that the Vandermonde code coefficients are also symmetric for odd $K$.
Explicitly, for odd $K$ and even $k$ we have
\begin{equation}
    \begin{aligned}
        \frac{1}{|c_{k}|^{2}} & = \left(\prod_{i\,\text{even} \neq k}|n_{k} - n_{i}|\right)\left(\prod_{i\,\text{odd}}\left|n_{k} - n_{i}\right|\right) \\
        & = \left(\prod_{j = 0, j \neq \frac{k}{2}}^{(K - 1)/2}\left|n_{k} - n_{2j}\right|\right)\left(\prod_{j = 0}^{(K - 1)/2}|n_{k} - n_{2j + 1}|\right) \\
        & = \left(\prod_{j = 0, j \neq \frac{k}{2}}^{(K - 1)/2}\left|\left(\frac{k}{2}S + n_{0}\right) - \left(jS + n_{0}\right)\right|\right)\left(\prod_{j = 0}^{(K - 1)/2}\left|\left(\frac{k}{2}S + n_{0}\right) - \left(jS + n_{1}\right)\right|\right) \\
        & = \prod_{j = 0, j \neq \frac{k}{2}}^{(K - 1)/2}\left|\left(\frac{K - 1}{2} - j\right)S + n_{1} - \left(\frac{K - 1}{2} - \frac{k}{2}\right)S - n_{1}\right| \times \\
        & \hspace{30pt} \prod_{j = 0}^{(K - 1)/2}\left|\left(\frac{K - 1}{2} - j\right)S + n_{0} - \left(\frac{K - 1}{2} - \frac{k}{2}\right)S - n_{1}\right| \\
        & = \left(\prod_{j' = 0, j' \neq \frac{K - k - 1}{2}}^{(K - 1)/2}\left|\left(j'S + n_{1}\right) - n_{K - k}\right|\right)\left(\prod_{j' = 0}^{(K - 1)/2}\left|\left(j'S + n_{0}\right) - n_{K - k}\right|\right) \\
        & = \left(\prod_{i'\,\text{odd} \neq K - k}|n_{i'} - n_{K - k}|\right)\left(\prod_{i'\,\text{even}}|n_{i'} - n_{K - k}|\right) \\
        & = \frac{1}{|c_{K - k}|^{2}} ~,
    \end{aligned}
\end{equation}
and likewise for odd $k$ by relabeling, since when $K$ is odd, every odd $k$ can be written as $k = K - k'$ for some even $k'$.
Thus, the Fock-grid coefficients of an odd $K$ Vandermonde code are symmetric, and since the Fock-space structure of a cyclic code is flip-symmetric for spin dimension $d = \frac{1}{2}(K - 1)S + n_{0} + n_{1} + 1$ when $K$ is odd, we can conclude that odd $K$ Vandermonde codes have dihedral symmetry, with $\hat{X}_{d}$ specifically swapping $|0_{L}\rangle$ and $|1_{L}\rangle$.
Of course, since only the magnitudes of the Vandermonde coefficients are constrained, this only holds if the phases of the $c_{k}$ are also chosen to be symmetric.

\section{Error detection protocols}\label{app:errorDetection}
Here we provide more details regarding the operations used in the error detection protocol described in Sec.~\ref{sec:errorDetection}, and present the corresponding protocols for the finite-support $(5, 0, 2)$, $K = 2$ and $(8, 0, 3)$, $K = 3$ codes.
For Vandermonde coefficients, the former code has logical and error words given by
\begin{equation}
    \begin{aligned}
        |0_{L}\rangle & = \sqrt{\frac{3}{5}}|0\rangle + \sqrt{\frac{2}{5}}|5\rangle ~, \\
        |1_{L}\rangle & = |2\rangle ~,
    \end{aligned}\hspace{40pt}
    \begin{aligned}
        |0_{\hat{a}}\rangle & = |4\rangle ~, \\
        |1_{\hat{a}}\rangle & = |1\rangle ~,
    \end{aligned}
    \label{eqn:025words}
\end{equation}
and the latter code has logical, first error, and second error words given by
\begin{equation}
    \begin{aligned}
        |0_{L}\rangle & = \sqrt{\frac{5}{16}}|0\rangle + \sqrt{\frac{11}{16}}|8\rangle ~, \\
        |1_{L}\rangle & = \sqrt{\frac{11}{16}}|3\rangle + \sqrt{\frac{5}{16}}|11\rangle ~,
    \end{aligned}\hspace{30pt}
    \begin{aligned}
        |0_{\hat{a}}\rangle & = |7\rangle ~, \\
        |1_{\hat{a}}\rangle & = \sqrt{\frac{3}{8}}|2\rangle + \sqrt{\frac{5}{8}}|10\rangle ~,
    \end{aligned}\hspace{30pt}
    \begin{aligned}
        |0_{\hat{a}^{2}}\rangle & = |6\rangle ~, \\
        |1_{\hat{a}^{2}}\rangle & = \sqrt{\frac{3}{28}}|1\rangle + \sqrt{\frac{25}{28}}|9\rangle ~.
    \end{aligned}
    \label{eqn:03811words}
\end{equation}

Combining Ramsey interferometry measurement (RIM)~\cite{sun_tracking_2014} with the selective parity manifold driving of Ref.~\cite{ni_beating_2023}, we consider the Hamiltonian
\begin{equation}
    \begin{aligned}
        H / \hbar & = -\chi \hat{n}|e\rangle\langle e| + H_{e} + H_{o} ~, \\
        H_{e/o} & = \Omega_{e/o}\sum_{n = 1}^{M}\left(e^{-i(\delta_{n, e/o}t + \phi)}|e\rangle\langle g| + \text{h.c.}\right) ~,
    \end{aligned}
\end{equation}
in the interaction picture, where $|g\rangle$ and $|e\rangle$ are the ground and excited states, respectively, of the ancilla qubit, which is dispersively coupled to the cavity housing the bosonic cyclic code.
Only allowing one of $\Omega_{e}$ and $\Omega_{o}$ to be nonzero at a time, setting the drive frequency detunings to $\delta_{n, e} = (2M - 2n)\chi$ and $\delta_{n, o} = (2M - 2n - 1)\chi$, and ensuring $2M\chi \gg \Omega_{e/o}$ as in~\cite{ni_beating_2023}, we can selectively drive the qubit about any equatorial axis for the even and odd Fock parity manifolds, respectively.
This also requires the highest occupied Fock state to be much less than $M$ so that the off-resonant detunings appear nearly symmetric to any occupied Fock state, yielding a cyclic evolution for any drive time $T$ satisfying $\chi T = m\pi$ for $m \in \mathbb{Z}$~\cite{ni_beating_2023}.
Thus, whenever the drive time obeys this condition, and $M$ is much greater than the highest occupied Fock state, we have the factor of $M$ in the condition $2M\chi \gg \Omega_{e/o}$, which allows for a faster Rabi rate.

Using this Hamiltonian, we can implement the three operations used in the protocol described in Sec.~\ref{sec:errorDetection}.
In particular, the idle evolution operations $U = \exp\left(it\hat{n}|e\rangle\langle e|\right)$ (blue) can be implemented by setting $\Omega_{e} = \Omega_{o} = 0$ and waiting for time $t' = \frac{t}{\chi}$.
The even/odd-select equatorial qubit rotations $R_{\cos(\phi)x + \sin(\phi)y}(\theta)$ (red/yellow) can be implemented by choosing $\Omega_{e/o}$ and drive time $T$ such that $\Omega_{e/o}T = \frac{\theta}{2}$ and $\chi T = m\pi$ for $m \in \mathbb{Z}$, and setting the other $\Omega_{o/e} = 0$.
By judiciously composing these three types of operations, we can move each relevant Fock state around its copy of the ancilla Bloch sphere, so that after each adaptive RIM step, the right Fock states end up in $|e\rangle$ to yield the desired POVM element.
In Fig.~\ref{fig:errorDetection}, we presented the protocol for the $(8, 0, 3)$, $K = 2$ code, which involved three relatively simple RIM steps.
In Fig.~\ref{fig:errorDetection025} and Fig.~\ref{fig:errorDetection03811}, we present the protocols for the $(5, 0, 2)$, $K = 2$ and $(8, 0, 3)$, $K = 3$ codes, respectively, which involve more complex RIM steps.
The latter protocol is particularly notable since it allows for non-destructive detection of up to two photon losses, which can also be theoretically exactly corrected if the code has Vandermonde coefficients.
We expect that similar protocols exist for other small finite-support cyclic codes, but they have to be worked out on a case-by-case basis (perhaps requiring some ingenuity) and may involve more complex RIM steps.

\begin{figure*}[ht]
    \centering
    \includegraphics[width=\textwidth]{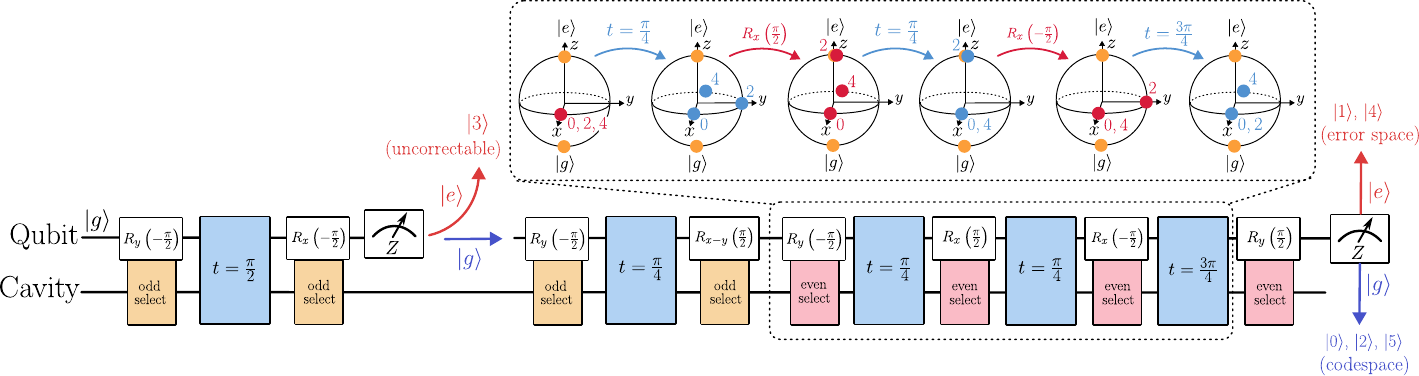}
    \caption{\textbf{Error detection protocol for finite-support $(5, 0, 2)$ code with $K = 2$.}
    This two-step adaptive RIM sequence involves the same types of operations as in Fig.~\ref{fig:errorDetection}, although with the evolution of the relevant Fock states around the ancilla Bloch sphere only shown for the even-select part of the second RIM step.
    A measurement of $|e\rangle$ after the first RIM step implies occupation of the $|3\rangle$ Fock state, corresponding to an uncorrectable error.
    Following instead a measurement of $|g\rangle$, the more involved second RIM step then distinguishes the codespace $\{|0\rangle, |2\rangle, |5\rangle\}$ from the error space $\{|1\rangle, |4\rangle\}$, accomplishing the desired POVM and error detection.
    Note that, as presented, in the second RIM step, the $|1\rangle$ state picks up a phase $-1$ relative to the $|4\rangle$ state, and the $|2\rangle$ state picks up a phase $i$ relative to the $|0\rangle$ and $|5\rangle$ states.
    The first of these can be corrected by adding a $t = \pi$ idle evolution $U = \exp\left(i\pi\hat{n}|e\rangle\langle e|\right)$ at the end of the protocol.
    The second can be corrected by adding an even-select $R_{x}\left(\pi\right)$ rotation, followed by a $t = \frac{3\pi}{4}$ idle evolution, followed by an even-select $R_{x}\left(-\pi\right)$ rotation at the end of the protocol.
    }
    \label{fig:errorDetection025}
\end{figure*}

\begin{figure*}[ht]
    \centering
    \includegraphics[width=\textwidth]{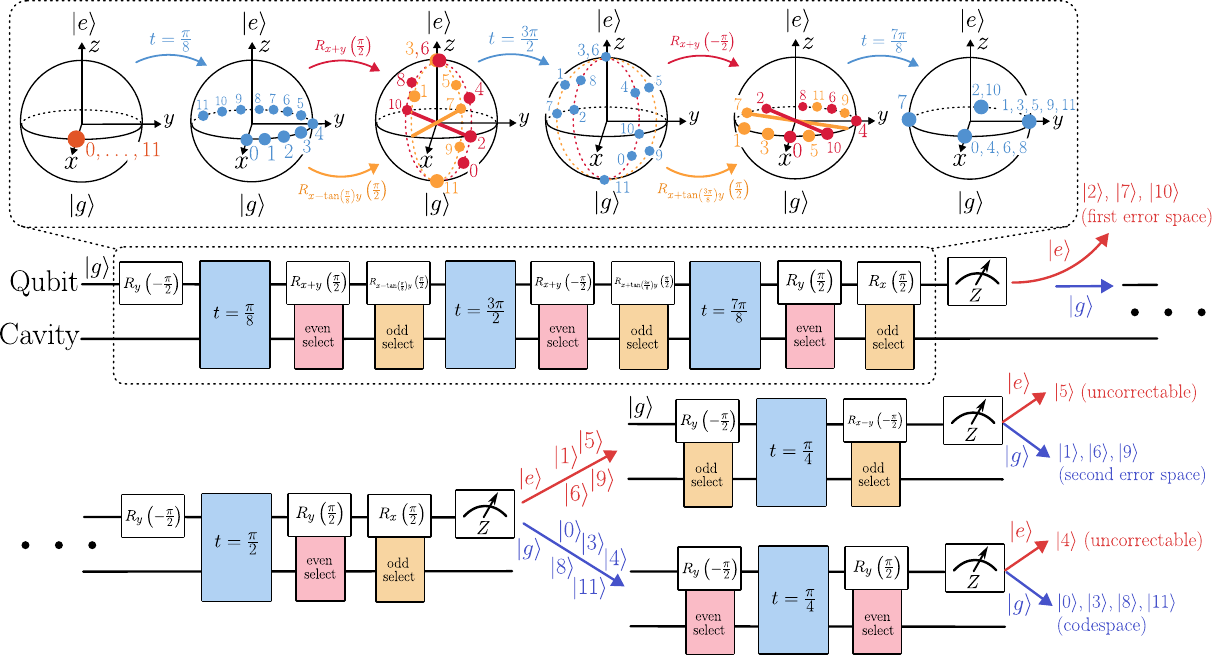}
    \caption{\textbf{Error detection protocol for finite-support $(8, 0, 3)$ code with $K = 3$.}
    This three-step adaptive RIM sequence involves the same types of operations as in Fig.~\ref{fig:errorDetection}, although with the evolution of the relevant Fock states around the ancilla Bloch sphere only shown for the first RIM step.
    A measurement of $|e\rangle$ after the first RIM step implies occupation of the first error space $\{|2\rangle, |7\rangle, |10\rangle\}$.
    Following instead a measurement of $|g\rangle$, the second RIM step then distinguishes the Fock state manifolds $\{|1\rangle, |5\rangle, |6\rangle, |9\rangle\}$ and $\{|0\rangle, |3\rangle, |4\rangle, |8\rangle, |11\rangle\}$, which are then further broken down in the final RIM step into an uncorrectable error $\{|5\rangle\}$ and the second error space $\{|1\rangle, |6\rangle, |9\rangle\}$, and an uncorrectable error $\{|4\rangle\}$ and the codespace $\{|0\rangle, |3\rangle, |8\rangle, |11\rangle\}$, respectively.
    This accomplishes the desired POVM, allowing for the detection of up to two photon loss errors.
    Note that, as presented, in the first RIM step, the $|2\rangle$ and $|7\rangle$ states pick up phases $-1$ and $i$, respectively, relative to the $|10\rangle$ state, and the $|3\rangle$ state picks up a phase $-1$ relative to the $|0\rangle$, $|8\rangle$, and $|11\rangle$ states.
    The first of these can be corrected by adding a $t = \frac{\pi}{8}$ idle evolution, followed by an even-select $R_{x}\left(\pi\right)$ rotation, followed by a $t = \frac{9\pi}{8}$ idle evolution, followed by an even-select $R_{x}\left(-\pi\right)$ rotation before the second RIM step.
    The second can be corrected by redoing the first RIM step, optionally dropping the even-select rotations, at the end of the protocol.
    The $|6\rangle$ state also picks up a phase $-1$ relative to the $|1\rangle$, $|5\rangle$, and $|9\rangle$ states during the first RIM step, but this phase cancels out during the second RIM step.
    }
    \label{fig:errorDetection03811}
\end{figure*}

\section{Multimode code extensions}\label{app:multimode}

\subsection{$\mathbb{Z}_{5}$ simplex code}\label{app:simplex}
In this section, we show how the $\mathbb{Z}_5$ simplex code introduced in Ref.~\cite{jain_quantum_2024} can be converted into a drum code, whose codewords consist of non-uniform superpositions of coherent states arranged as regular prisms in the coherent-state parameter space\footnote{The term “drum” refers to the geometric resemblance between these regular prism structures and the shape of a drum.}, that admits a Gaussian implementation of the $\frac{2\pi}{5}$-phase gate.
This is achieved by selecting logical codewords from different stabilizer-syndrome sectors, at the cost of reducing the loss distance by one.

The original codewords of the $\mathbb{Z}_5$ simplex code are defined as
\begin{equation}
\begin{aligned}
    \ket{\overline{+}}&\propto \sum_{j=0}^{4} \ket{\omega^j \alpha, \omega^{2j} \alpha},\\
    \ket{\overline{-}}&\propto \sum_{j=0}^{4} \ket{-\omega^j \alpha, -\omega^{2j} \alpha},\\
    \ket{\overline{0}}&\propto \sum_{\substack{n,m\geq 0\\ n+2m \mod 5=0\\ n+m ~\text{is even}}} \frac{\alpha^{n+m}}{\sqrt{n!m!}} \ket{n,m},\\
    \ket{\overline{1}}&\propto \sum_{\substack{n,m\geq 0\\ n+2m \mod 5=0\\ n+m ~\text{is odd}}} \frac{\alpha^{n+m}}{\sqrt{n!m!}} \ket{n,m},
    \end{aligned}
\end{equation}
where $\omega=e^{i 2\pi/5}$ is the fifth root of unity and $\alpha$ is chosen to be a real amplitude for convenience. 

This code has stabilizers 
\begin{equation}
    \{\omega^{\hat{a}_1^\dagger \hat{a}_1+ 2\hat{a}_2^\dagger \hat{a}_2}, \hat{a}_1^2 \hat{a}_2^4 -\alpha^6, \hat{a}_1^3 \hat{a}_2-\alpha^4\},
\end{equation}
and logical operators $\bar{Z}=(-1)^{\hat{a}_1^\dagger \hat{a}_1 +\hat{a}_2^\dagger \hat{a}_2},\,\bar{X}= \frac{\hat{a}_1 \hat{a}_2^2}{\alpha^3}$.
Note that the minimum degree loss monomial that yields a logical $\bar{X}$ error is 3, so this code can correct single-photon loss in either of the modes. 

The $\mathbb{Z}_5$ modular stabilizer $\hat{S}_{5}=\omega^{\hat{a}_1^\dagger \hat{a}_1 + 2\hat{a}_2^\dagger \hat{a}_2}$ imposes the constraint $\hat{n}_1 +2\hat{n}_2 \mod 5 \equiv 0$.
Single-photon losses $\hat{a}_1$ and $\hat{a}_2$ yield distinct syndromes 
\begin{equation}
\begin{aligned}
    \hat{S}_{5} \hat{a}_1 \ket{\psi}&= \alpha \omega^4 \ket{\psi},\\
    \hat{S}_{5} \hat{a}_2 \ket{\psi}&= \alpha \omega^3 \ket{\psi},
    \end{aligned}
\end{equation}
for any $\ket{\psi}=\alpha \ket{+_{L}} + \beta \ket{-_{L}}$ in the codespace.
These distinct syndromes allow the two loss processes to be unambiguously identified and corrected by measuring the stabilizer $\hat{S}_{5}$.

The codewords of the simplex code form a spherical $2$-design \cite{delsarte1991spherical,roy2014complex,mohammadpour2024complex}, implying that all loss–gain monomials up to degree two are automatically detectable.
We verify some of the Knill–Laflamme conditions explicitly,
\begin{equation}
    \begin{aligned}
        \langle 0_{L}| \hat{a}_1^\dagger \hat{a}_1 |0_{L} \rangle&= \langle 0_{L}| \hat{a}_2^\dagger \hat{a}_2 |0_{L} \rangle= \langle 1_{L}| \hat{a}_1^\dagger \hat{a}_1 |1_{L} \rangle=
        \langle 1_{L}| \hat{a}_2^\dagger \hat{a}_2 |1_{L} \rangle=\alpha^2,\\
        \langle 0_{L}| \hat{a}_1^\dagger \hat{a}_2 \ket{1_{L}}&=\langle 0_{L}| \hat{a}_1^\dagger \hat{a}_1 |1_{L} \rangle = \langle 0_{L}| \hat{a}_2^\dagger \hat{a}_2 | 1_{L} \rangle =0,\\
        \langle \pm_{L}| \mathbf{\hat{a}}^{\dagger \mathbf{p}} \mathbf{\hat{a}}^{ \mathbf{p}}| \pm_{L} \rangle &\sim \frac{\alpha^{2|\mathbf{p}|}}{5} \exp(-\alpha^2 d_E ),
    \end{aligned}
\end{equation}
where $d_E=1.5$ is the minimum Euclidean distance~\cite{jain_quantum_2024} between coherent-state points in the $\ket{\pm_{L}}$ code constellations, and $\mathbf{\hat{a}}^{\dagger \mathbf{p}} \mathbf{\hat{a}}^{\mathbf{p}} \equiv \hat{a}_1^{\dagger p_1} \hat{a}_2^{\dagger p_2} \hat{a}_1^{p_1} \hat{a}_2^{p_2}$ is short-hand notation for a dephasing (balanced loss--gain) error, written using the multi-index $\mathbf{p} = (p_1, p_2)$ with operator vector $\mathbf{\hat{a}} = (\hat{a}_1, \hat{a}_2)$ and modulus $|\mathbf{p}| = p_1 + p_2$.
As can be seen, this code can exactly correct single-photon losses $\hat{a}_1, \hat{a}_2$ and exponentially suppress dephasing error $\mathbf{\hat{a}}^{\dagger \mathbf{p}} \mathbf{\hat{a}}^{ \mathbf{p}}$.

We observe that the syndrome of stabilizer $\hat{S}_{5}$ can take values from integer multiples of $\frac{2\pi}{5}$.
This allows us to construct a drum code by selecting logical states from different syndrome sectors of
$\hat{S}_{5}$ to define the new logical $|0/1_{L}\rangle$ states such that the original stabilizer $\hat{S}_{5}$ becomes a logical phase gate.
The drum codewords are defined as
\begin{equation}
    \begin{aligned}
        \ket{0_{L}}_{\text{drum}}&=\ket{0_{L}} ~, \\
        \ket{1_{L}}_{\text{drum}}&=\frac{\hat{a}_2}{\alpha}\ket{1_{L}} ~,
    \end{aligned}
\end{equation}
where the logical $\ket{0_{L}}_{\text{drum}}$ is the same as the one for the $\mathbb{Z}_5$ simplex code, and the logical $\ket{1_{L}}_{\text{drum}}$ is obtained by applying a loss to the second mode to the logical $|1_{L}\rangle$ state of the simplex code.

The modified codewords are stabilized by $\hat{a}_1^2 \hat{a}_2^4, \hat{a}_1^3 \hat{a}_2$, in particular, the adjoint of the original modular stabilizer $\hat{S}_{5}^\dagger=\omega^{-\hat{a}_1^\dagger \hat{a}_1 - 2\hat{a}_2^\dagger \hat{a}_2}$ now acts as a logical $\frac{4\pi}{5}$-phase gate
\begin{equation}
    \begin{aligned}
        \hat{S}_{5}^\dagger \ket{0_{L}}_{\text{drum}}&=\ket{0_{L}},\\
        \hat{S}_{5}^\dagger \ket{1_{L}}_{\text{drum}}&= \hat{S}_{5}^\dagger \frac{\hat{a}_2}{\alpha}\ket{1_{L}}=\omega^2 \ket{1_{L}}_{\text{drum}}.
    \end{aligned}
\end{equation} 
The logical $X$ operator is now $\bar{X}_{\text{drum}}=\frac{\hat{a}_1 \hat{a}_2}{\alpha^2}$.

Since the logical $|0/1_{L}\rangle$ states of the drum code are in the different syndrome sectors of the $\mathbb{Z}_5$ simplex code, the drum code can no longer correct single-photon losses but can still detect them
\begin{equation}
\begin{aligned}
    \langle 0_{L}| \hat{a}_1 | 0_{L} \rangle_{\text{drum}}&= \langle 0_{L}| \hat{a}_2 | 0_{L} \rangle_{\text{drum}}=  \langle 1_{L}| \hat{a}_1 | 1_{L} \rangle_{\text{drum}}=  \langle 1_{L}| \hat{a}_2 | 1_{L} \rangle_{\text{drum}}=\alpha,\\
    \langle 0_{L} | \hat{a}_1 | 1_{L} \rangle_{\text{drum}}& = \langle 0_{L} | \hat{a}_2 | 1_{L} \rangle_{\text{drum}}=0,\\
    \langle \pm_L | \mathbf{\hat{a}}^{\dagger \mathbf{p}} \mathbf{\hat{a}}^{ \mathbf{p}}| \pm_{L} \rangle_{\text{drum}} &\sim \frac{\alpha^{2|\mathbf{p}|}}{5} \exp(-\alpha^2 d_E ).
    \end{aligned}
\end{equation}
The asymptotic dephasing suppression is unchanged since it is completely determined by the minimum Euclidean distance between the constellations of the $|\pm_{L}\rangle$ codewords, which is the same as the one for the simplex code.
In summary, by reducing the loss distance by one, we obtain a drum code that admits a Gaussian implementation of the $\frac{2\pi}{5}$-phase gate while retaining dephasing suppression and loss detectability.

\subsection{Octahedron drum code}
In this section, we introduce a code derived from an octahedral constellation and show how it can be converted into a drum code that supports a Gaussian $\frac{2\pi}{3}$-phase gate. 

We begin with an octahedral-constellation code whose logical states are
\begin{equation}
\begin{aligned}
    \ket{+_{L}}& \propto \sum_{j=0}^2 \ket{\omega^j \alpha_1, \alpha_2} + \ket{-\omega^j \alpha_1, -\alpha_2},\\
    \ket{-_{L}}& \propto \sum_{j=0}^2 \ket{-\omega^j \alpha_1, \alpha_2} + \ket{\omega^j \alpha_1, -\alpha_2},\\
    \ket{0_{L}}&\propto \sum_{\substack{n,m\geq 0\\ n \mod 3=0\\ n, m ~\text{are even}}} \frac{\alpha_1^n \alpha_2^m}{\sqrt{n!m!}} \ket{n,m},\\
    \ket{1_{L}}&\propto \sum_{\substack{n,m\geq 0\\ n \mod 3=0\\ n, m ~\text{are odd}}} \frac{\alpha_1^n \alpha_2^m}{\sqrt{n!m!}} \ket{n,m},
    \end{aligned}
\end{equation}
where $\omega=e^{i \frac{2\pi}{3}}$,$\alpha_1=\sqrt{\frac{2}{3}} \alpha, \alpha_2=\sqrt{\frac{1}{3}} \alpha$ and $\alpha_1^2+\alpha_2^2=\alpha^2$. 

The octahedron code is stabilized by 
\begin{equation}
    \{\omega^{\hat{a}_1^\dagger \hat{a}_1}, (-1)^{\hat{a}_1^\dagger \hat{a}_1 + \hat{a}_2^\dagger \hat{a}_2} ,\hat{a}_1^6-\alpha_1^6, \hat{a}_2^2-\alpha_2^2\}, 
\end{equation}
with logical operators $\bar{Z}=(-1)^{\hat{a}_2^\dagger \hat{a}_2}, \bar{X}=\frac{\hat{a}_1^3 \hat{a}_2}{\alpha_1^3 \alpha_2}$.
Because the logical states $\ket{0_{L}}, \ket{1_{L}}$ occupy distinct sectors of Fock space, namely $\{\ket{6p, 2q }\}_{p,q \geq 0}$ and $\{\ket{6p+3, 2q+1 }\}_{p,q \geq 0}$, any loss monomial of degree up to three is detectable.
In particular, single-photon loss in either mode is correctable.
We can verify loss operators $\hat{a}_1$ and $\hat{a}_2$ have distinct syndromes for modular stabilizers 
\begin{equation}
    \hat{S}_3=\omega^{\hat{a}_1^\dagger \hat{a}_1} ~~\text{and}~~ \hat{S}_2=(-1)^{\hat{a}_1^\dagger \hat{a}_1 +\hat{a}_2^\dagger \hat{a}_2}.
\end{equation} 
Table~\ref{tab:drum_syndrome} lists the syndromes of all loss monomials up to degree three.

\begin{table}[h]
    \centering
    \begin{tabular}{|c|c|c|}
    \hline
         & $\hat{S}_3$ syndrome & $\hat{S}_2$ syndrome \\
         \hline
         $\hat{a}_1$&  $\omega^2$ & $-1$  \\
         $\hat{a}_1^2$&  $\omega$ & $1$  \\
         $\hat{a}_1^3$&  $1$ & $-1$  \\
         $\hat{a}_2$ &  $1$ & $-1$ \\
         $\hat{a}_2^3$ &  $1$ & $-1$ \\
         $\hat{a}_1 \hat{a}_2$ & $\omega^2$ & $1$\\
         $\hat{a}_1^2 \hat{a}_2$ & $\omega$ & $-1$\\
         $\hat{a}_1 \hat{a}_2^2$ & $\omega^2$ & $-1$\\
         \hline
    \end{tabular}
    \caption{Syndrome table of all loss monomial with degree less than or equal to 3.
    Here we skip the dissipator $\hat{a}_2^2$ since it is proportional to logical identity and has trivial syndromes. }
    \label{tab:drum_syndrome}
\end{table}

The codewords of the octahedron code form a spherical $3$-design \cite{delsarte1991spherical,roy2014complex,mohammadpour2024complex}, implying that all loss–gain monomials up to degree three are automatically detectable.
We therefore omit a detailed Knill–Laflamme analysis.
Dephasing errors are exponentially suppressed, with
\begin{equation}
    \langle \mp_{L}| \mathbf{\hat{a}}^{\dagger \mathbf{p}} \mathbf{\hat{a}}^{ \mathbf{p}}| \pm_{L} \rangle \sim \frac{\alpha^{2|\mathbf{p}|}}{6} \exp(-\alpha^2 d_E ),
\end{equation}
where $d_E=\frac{2}{3}$ is the minimum Euclidean distance between coherent-state points in the $\ket{\pm_{L}}$ code constellations, with $\mathbf{\hat{a}}^{\dagger \mathbf{p}} \mathbf{\hat{a}}^{\mathbf{p}}$ the multi-index dephasing notation introduced in App.~\ref{app:simplex}.

We observe that the modular stabilizer $\hat{S}_3$ has eigenvalues from $\{1, \omega, \omega^2\} $ where $\omega$ is the third-root of unity.
This allows us to select codewords from different syndrome sectors of $\hat{S}_3$ and define a drum code that admits $\frac{2\pi}{3}$-phase gate while retaining single-photon loss detectability.
The drum codewords are defined as
\begin{equation}
    \begin{aligned}
        \ket{0_{L}}_{\text{drum}} & =\ket{0_{L}},\\
        \ket{1_{L}}_{\text{drum}} & = \frac{\hat{a}_1}{\alpha_1}\ket{1_{L}}
    \end{aligned}
\end{equation}
By applying the adjoint of the modular stabilizer $\hat{S}_3$ to the drum codewords, we get
\begin{equation}
    \begin{aligned}
        \hat{S}_3^\dagger \ket{0_{L}}_{\text{drum}}&= \ket{0_{L}}_{\text{drum}},\\
        \hat{S}_3^\dagger \ket{1_{L}}_{\text{drum}}&= \omega\ket{1_{L}}_{\text{drum}},
    \end{aligned}
\end{equation}
showing that the previous stabilizer $\hat{S}_3^\dagger$ is now a logical $\frac{2\pi}{3}$-phase gate of the drum code.
The logical $\bar{X}$ operator of the drum code is $\bar{X}_{\text{drum}}=\frac{\hat{a}_1^2 \hat{a}_2}{\alpha_1^2 \alpha_2}$.
Notably, most error-correction properties of the octahedron code are inherited by the drum code: it can detect all degree-2 loss monomials and exhibits the same asymptotic suppression of dephasing errors.

\twocolumngrid
\bibliography{vf-references.bib, ow-references.bib, yx-references.bib}

\end{document}